\documentclass[reqno]{amsart}
\usepackage{amsfonts}%[12pt]或者[reqno]
\usepackage{mathrsfs}

\usepackage{cite}
\usepackage{amscd}
\usepackage{latexsym}
\usepackage{amsfonts}
\usepackage{graphicx}
\usepackage{CJK,indentfirst,amsmath, amsfonts,amssymb,amsthm,cite,cases, subeqnarray,setspace,multirow}

\usepackage[colorlinks, linkcolor=blue, citecolor=blue]{hyperref}

\begin{document}
%\begin{CJK*}{GBK}{song}

\title[]
{Quasineutral limit for the quantum Navier-Stokes-Poisson equation}

\author[M. Li, X. Pu and S. Wang]{Min Li,\ \ \  Xueke Pu\ \ \ and\ \ \  Shu Wang}  % in alphabetical order

\address{Min Li \newline
School of Mathematics and statistics, Chongqing University, Chongqing 401331, P.R.China} \email{530294442@qq.com}

\address{Xueke Pu \newline
Department of Mathematics, Chongqing University, Chongqing 401331, P.R.China} \email{ xuekepu@cqu.edu.cn}

\address{Shu Wang \newline
College of Applied Sciences, Beijing University of Technology, Beijing 100124, People's Republic of China
}
\email{wangshu@bjut.edu.cn}

\thanks{The second author is supported in part by NSFC (11471057) and the Fundamental Research Funds for the Central Universities (Project No. 106112016CDJZR105501).The third author is supported by NSFC (11371042) and the key fundation of Beijing Municipal Education Commission.}
 \keywords{}
% 换行用此命令，放在括号里\hfill\break\indent }

\begin{abstract}

In this paper, we study the quasineutral limit and asymptotic behaviors for the quantum Navier-Stokes-Possion equation. We apply a formal expansion according to Debye length and derive the neutral incompressible Navier-Stokes equation. To establish this limit mathematically rigorously, we derive uniform (in Debye length) estimates for the remainders, for well-prepared initial data. It is demonstrated that the quantum effect do play important roles in the estimates and the norm introduced depends on the Planck constant $\hbar>0$.

\end{abstract}

\maketitle \numberwithin{equation}{section}
\newtheorem{proposition}{Proposition}[section]
\newtheorem{theorem}{Theorem}[section]
\newtheorem{lemma}[theorem]{Lemma}
\newtheorem{remark}[theorem]{Remark}
\newtheorem{hypothesis}[theorem]{Hypothesis}
\newtheorem{definition}{Definition}[section]
\newtheorem{corollary}{Corollary}[section]
\newtheorem{assumption}{Assumption}[section]

\section{Introduction}
\setcounter{section}{1}\setcounter{equation}{0}
The purpose of this paper is to consider the following hydrodynamic Navier-Stokes-Possion equation, describing the motion of the electrons in plasmas with the electric potential \cite{Gardner94},
\begin{subequations}\label{equ1}
\begin{numcases}{}
\frac{\partial n}{\partial t}+\mathrm{div}(nu)=0,\label{e1n}\\
\frac{\partial (nu)}{\partial t}+\mathrm{div}(nu\otimes u-P)+n\nabla\mathrm{V}=\mathrm{div}S,\label{e1u}\\
\frac{\partial W}{\partial t}+\mathrm{div}(uW-uP+q)+nu\cdot\nabla\mathrm{V}=\mathrm{div}(uS),\label{e1t}\\
-\varepsilon\Delta V=n-1,
\end{numcases}
\end{subequations}
where $V=-e\phi$ is the electrostatic potential, $n$ is the electron density,  $u=(u_1,u_2,u_3)$ is the velocity, $nu$ is the momentum density, $P=(P_{ij})_{3\times3}$ is the stress tensor, $W$ is the energy density and $q$ is the heat flux. It is customary that the heat flux is assumed to obey the Fourier law $q=-\kappa\nabla T$. To close the moment expansion at the third order, we define the above stress tensor $P$ and $W$ in terms of the density $n,u$ and $T$ by
\[
P_{ij}=-nT\delta_{ij}+\frac{\hbar^2n}{12}\frac{\partial^2}{\partial x_i\partial x_j}\log n,
\]
and
\[
W=\frac{3}{2}nT+\frac12n|u|^2-\frac{\hbar^2n}{24m}\Delta\log n,
\]
respectively, where $\hbar>0$ is the Planck constant and is small compared to macro quantities. We have omitted the terms that are $O(\hbar^4)$. The viscous stress tensor $S$ is given by
\[
S=\mu(\nabla u+(\nabla u)^{\top})+\lambda(\mathrm{div}u)I,
\]
to include the viscosity, with $\mu>0$ and $2\mu+3\lambda\geq0$. The quantum stress tensor is closely related to the quantum Bohm potential \cite{Wigner32,Bohm52}
\[
Q(n)=-\frac{\hbar^2}{2}\frac{\Delta\sqrt{n}}{\sqrt{n}},
\]
through the formula
$$-\frac{\partial}{\partial x_i}P_{ij}=\frac{\partial}{\partial x_i}(nT)+\frac{n}{3}\frac{\partial Q}{\partial x_i}.$$
By simple calculations, it gives that
\[
-n\nabla Q(n)=\frac{\hbar^2}{4}\Delta\nabla n -\frac{\hbar^2}{4}\{\frac{(\Delta n\nabla n+\nabla n\cdot\nabla^{2}n)}{n}-\frac{(\nabla n\cdot\nabla n)\nabla n}{n^{2}}\}.
\]

In the above equation \eqref{equ1}, $\varepsilon$ is the scaled squared Debye length, which vary by many orders of magnitude. Typical values of the Debye length go from $10^{-3}m$ to $10^{-8}m$. The Debye length is a fundamental and important parameter in plasmas, below which charge separation occurs. In practical applications, the Debye length is small compared to characteristic observation length and hence is interesting to study the limit as $\varepsilon$ goes to zero. Formally, this will leads to an incompressible Navier-Stokes equation (See Sect. \ref{Sect1.1} below) and this paper aims to study this limit mathematically rigorously.

Before we proceed, we recall some derivations for the quantum fluid dynamics equation. The full compressible quantum Navier-Stokes-Poisson (FCQNSP) model \eqref{equ1} has many practical applications in fluid models of nucleus, superconductivity and superfluidity and semiconductor devices, and was derived by Gardner \cite{Gardner94} from a moment expansion of the Wigner-Boltzmann equation. When there is no electrostatic potential $V$, it reduces to the full compressible quantum Navier-Stokes (FCQNS) equation, which was also derived by J\"{u}ngel and Mili\^{s}i\'{c} \cite{JM11} from the collisional Wigner equation \cite{Wigner32} by performing a Chapman-Enskog expansion. See also the references therein for the derivation of other quantum hydrodynamics equations from Wigner equation. In particular, Degond and Ringhofer \cite{DR03} formally derived nonlocal quantum hydrodynamic models, where the quantum stress tensor may by non-diagonal and the quantum heat flux many not vanish, in contrast to the classical equations. For more models taking quantum effects into consideration, see also \cite{DGMR06,DGMR08,DMR05,DGM07b}. By performing the zero-mean-free-path limit in the collisional Wigner equation, Brull and M\'{e}hats derived the isothermal quantum Navier-Stokes equation \cite{BM10}.

Although have many applications in various fields of physics, the full quantum Navier-Stokes-Poisson equation \eqref{equ1} and the full quantum Navier-Stokes equation are less studied mathematically rigorously, to our best knowledge. Without quantum effects, they are all comprehensively studied in various aspects of mathematics, see for example \cite{Temam01,Lions96,CF88,Stein1970}. For the full compressible quantum Navier-Stokes equation, the only result known to the authors is the global existence of small classical solutions in $R^3$ when the viscosity and heat conductivity are present, obtained recently in \cite{PGG15} following the seminal paper of Matsumura and Nishida \cite{MN80}. For other results on the quantum hydrodynamic equations and related models, the readers may refer to \cite{Haas11,JLW14,Jungel10,DS85,LL05,HL94,HL96} and the references therein.

Next, we recall some mathematical results on quasi-neutral limit for various hydrodynamic equation. Indeed, in the recent two decades, the quasi-neutral limit problem has attracted many attentions of physicists and applied mathematicians. To the author's best knowledge, the first quasi-neutral limit is on the Euler-Poisson equation for ions with positive ion temperature by Cordier and Grenier \cite{CG00}. This result was recently generalized to the Euler-Poisson equation with zero ion temperature for cold plasmas in \cite{PGQ15}. Since then, quasi-neutral limit results have been obtained for various hydrodynamic models in plasmas. See \cite{Wang04} for the Euler/Navier-Stokes-Poisson system with and without viscosity, \cite{WJ06} for the Navier-Stokes-Poisson equation, \cite{PWY06} for the Cauchy problem for the non-isentropic Euler-Poisson equation with prepared initial data, \cite{JLL09} for the non-isentropic compressible Navier-Stokes-Poisson, \cite{DM12} for the analysis of oscillations and defect measures, to list only a few. See also \cite{DM08,GM01,GM2001,CDM13,DM15} for the other related limits. For hydrodynamic models with capillarity or Korteweg effects, there are also some quasineutral limit results, see \cite{LY14,BDD05} for example. There is a vast literature concerning the quasi-neutral limit for various models, and we cannot give a complete list here. The interested readers may refer to these papers above and the references therein.

As pointed out above, we aim to study the quasineutral limit mathematically rigorously. To be precise, we confine ourselves to the Cauchy problem for \eqref{equ1} in $\Bbb R^3$. We show that as Debye length goes to zero, the smooth solutions converges to solutions of the incompressible Navier-Stokes equation \eqref{e0}, at least for well-prepared initial data (The case for ill-prepared initial data will be treated somewhere else). We also obtained convergence rate w.r.t. the Debye length parameter $\varepsilon$. The main result is stated in Theorem \ref{thm3}. Due to the special structure of \eqref{equ1}, to get uniform in $\varepsilon$ estimates for the remainder terms we need to carefully use the structure of the equation and construct suitable energy norms in estimates. The norm we finally adopt is the triple norm defined in \eqref{def-A} incorporating the quantum parameter $\hbar>0$. Since higher order terms appear in this triple norm, much effort is needed to close the estimate in the proof.

In the rest of Introduction, we first give the formal expansions and derive the incompressible Navier-Stokes equation \eqref{e0} for the leading terms and then we derive the remainder equation \eqref{rem1} and state the main result in Theorem \ref{thm3}.

\subsection{Formal Expansions}\label{Sect1.1}
In terms of $(n,u,T,\phi)$, the above quantum Navier-Stokes-Poisson system \eqref{equ1} with the electric potential can be rewritten as
\begin{subequations}\label{equ2}
\begin{numcases}{}
\partial_tn+\nabla\cdot(n{u})=0,\\
\partial_t{u}+{u}\cdot\nabla{u}+\frac{1}{n}\nabla(nT)-\frac{\hbar^2}{12n}\mathrm{div}\{n(\nabla\otimes\nabla)\log n\}-\nabla\phi-\frac{\mu}{n}\Delta u\nonumber\\
\ \ \ \ \ \ \ -\frac{\mu+\lambda}{n}\nabla\mathrm{div}u=0,\\
\partial_tT+{u}\cdot\nabla T+\frac23T\nabla\cdot{u}-\frac2{3n}\nabla\cdot(\kappa\nabla T)+\frac{\hbar^2}{36n}\nabla\cdot(n\Delta{u}) \notag\\
\ \ \ \ \ \ \ -\frac{2}{3n}\Big\{\frac{\mu}{2}|\nabla u+(\nabla u)^{\top}|^{2} +\lambda(\mathrm{div}u)^{2}\Big\}=0,\\
\varepsilon\Delta\phi=n-1.
\end{numcases}
\end{subequations}

Plugging the following formal expansion
\begin{subequations}\label{expan-formal}
\begin{numcases}{}
n=n^{(0)}+\varepsilon^{1} n^{(1)}+\varepsilon^{2} n^{(2)}+\varepsilon^{3} n^{(3)}+\varepsilon^{4} n^{(4)}+\cdots,\\
u=u^{(0)}+\varepsilon^{1} u^{(1)}+\varepsilon^{2} u^{(2)}+\varepsilon^{3}u^{(3)}+\varepsilon^{4} u^{(4)}+\cdots,\\
T=T^{(0)}+\varepsilon^{1} T^{(1)}+\varepsilon^{2} T^{(2)}+\varepsilon^{3}T^{(3)}+\varepsilon^{4} T^{(4)}+\cdots,\\
\phi=\phi^{(0)}+\varepsilon^{1}\phi^{(1)}+\varepsilon^{2}\phi^{(2)} +\varepsilon^{3}\phi^{(3)}+\varepsilon^{4} \phi^{(4)}+\cdots,
\end{numcases}
\end{subequations}
into \eqref{equ2}, we get a power series of $\varepsilon$. At $O(1)$ order, we obtain $n^{(0)}=1$ and the system for $(u^{(0)},T^{(0)},\phi^{(0)})$,
\begin{subequations}\label{e0}
\begin{numcases}{(\mathcal S_0)}
\mathrm{div}u^{(0)}=0,\\
\partial_t{u^{(0)}}+{u^{(0)}}\cdot\nabla{u^{(0)}}+\nabla T^{(0)}-\nabla\phi^{(0)}-\mu\Delta u^{(0)}=0,\\
\partial_tT^{(0)}+{u^{(0)}}\cdot\nabla T^{(0)}-\frac{2\kappa}{3}\Delta T^{(0)}-\frac{\mu}{3}|\nabla u^{(0)}+(\nabla u^{(0)})^{\top}|^{2}=0.
\end{numcases}
\end{subequations}
Formally, as $\varepsilon\to0$, we known that the solutions of \eqref{equ2} should converge to those of \eqref{e0}. For the solvability of $(u^{(0)},T^{(0)},\phi^{(0)})$, we have the following theorem. See \cite{Lions96} for details.
\begin{theorem}\label{thm1}
Let $\tilde s\geq\frac N2+1$. Then for any given initial data $(u^{(0)}_{0},T^{(0)}_{0},\phi^{(0)}_{0})\in H^{\tilde s+3}$ and $\mathrm{div}u^{(0)}_{0}=0$, there exists some $\tau_*>0$ such that the initial value problem \eqref{e0} has a unique solution such that for any $\tau<\tau_*$, the following holds
\begin{equation}
\begin{split}
\sup_{t\in[0,\tau]}\|(u^{(0)},T^{(0)},\phi^{(0)})\|_{H^{\tilde{s}+3}}&+\|(u^{(0)},T^{(0)})\|_{L^{2}(0,\tau;{H^{\tilde{s}+4}})}\\&\leq C\|(u^{(0)},T^{(0)},\phi^{(0)})(0)\|_{H^{\tilde{s}+3}},
\end{split}
\end{equation}
where $C$ is a constant.
\end{theorem}

Next, we need to find out the equations satisfied by $(n^{(1)},u^{(1)},T^{(1)},\phi^{(1)})$. To this end,  taking $(n^{(0)},u^{(0)},T^{(0)},\phi^{(0)})$ as known functions and at order $O(\varepsilon)$, we obtain
\begin{subequations}\label{j1}
\begin{numcases}{(\mathcal S_1)}
\mathrm{div}u^{(1)}=-\partial_{t}\Delta\phi^{(0)} -u^{(0)}\cdot\nabla\Delta\phi^{(0)},\\
\partial_t{u^{(1)}}+{u^{(0)}}\cdot\nabla{u^{(1)}} +{u^{(1)}}\cdot\nabla{u^{(0)}}+\nabla T^{(1)}-\nabla\phi^{(1)}-\mu\Delta u^{(1)}=-f_{0},\\
\partial_tT^{(1)}+{u^{(1)}}\cdot\nabla T^{(0)}+{u^{(0)}}\cdot\nabla T^{(1)}-\frac{2\kappa}{3}\Delta T^{(1)}\notag\\\ \ \ \ \ \ \ -\frac{2\mu}{3}(\nabla u^{(0)}+(\nabla u^{(0)})^{\top})(\nabla u^{(1)}+(\nabla u^{(1)})^{\top})=-g_{0},\\
\Delta\phi^{(0)}=n^{(1)},
\end{numcases}
\end{subequations}
where $f_{0}$ and $g_{0}$ depend only on $(u^{(0)},T^{(0)},\phi^{(0)})$ and are given by the following
\begin{subequations}\label{def1}
\begin{numcases}{}
f_{0}=T^{(0)}\nabla\Delta\phi^{(0)}-\frac{\hbar^{2}}{12}\nabla\Delta\Delta\phi^{(0)}+\mu(\Delta\phi^{(0)}\Delta u^{(0)})+(\mu+\lambda)\nabla(\partial_{t}\Delta\phi^{(0)}\notag\\\ \ \ \ \ \ \  +u^{(0)}\cdot\nabla\Delta\phi^{(0)}),\\
g_{0}=\frac 23T^{(0)}(-\partial_{t}\Delta\phi^{(0)}-u^{(0)}\cdot\nabla\Delta\phi^{(0)})+\frac{2\kappa}{3}\Delta\phi^{(0)}\Delta T^{(0)}\notag\\\ \ \ \ \ \ \ +\frac{\hbar^{2}}{36}\Big\{\Delta(-\partial_{t}\Delta\phi^{(0)} -u^{(0)}\cdot\nabla\Delta\phi^{(0)})+\nabla\Delta\phi^{(0)}\cdot\Delta u^{(0)}\Big\}\notag\\\ \ \ \ \ \ \ +\frac{2\mu}{3}\Delta\phi^{(0)}(\frac 12|\nabla u^{(0)}+(\nabla u^{(0)})^{\top}|^{2}).
\end{numcases}
\end{subequations}

Similarly, at the $O(\varepsilon^2)$ order, we obtain for $(n^{(2)},u^{(2)},T^{(2)},\phi^{(2)})$ that
\begin{subequations}\label{j2}
\begin{numcases}{(\mathcal S_2)}
\mathrm{div}u^{(2)}=-\partial_{t}\Delta\phi^{(1)} -\mathrm{div}(\Delta\phi^{(1)}u^{(1)} +\Delta\phi^{(1)}u^{(0)}),\label{j21}\\
\partial_t{u^{(2)}}+{u^{(0)}}\cdot\nabla{u^{(2)}} +{u^{(2)}}\cdot\nabla{u^{(0)}}+\nabla T^{(2)}-\nabla\phi^{(2)}-\mu\Delta u^{(2)}\notag\\ \ \ \ \ \ \ \ \ \ \ =-f_{1},\label{j22}\ \ \ \ \ \ \ \ \ \ \\
\partial_tT^{(2)}+{u^{(0)}}\cdot\nabla T^{(2)}+{u^{(2)}}\cdot\nabla T^{(0)}-\frac{2\kappa}{3}\Delta T^{(2)}\notag\\ \ \ \ \ \ \ \ \ \ \ -\frac{2\mu}{3}\{(\nabla u^{(0)}+(\nabla u^{(0)})^{\top})(\nabla u^{(2)}+(\nabla u^{(2)})^{\top})=-g_{1}\label{j23},\\
\Delta\phi^{(1)}=n^{(2)},
\end{numcases}
\end{subequations}
where $f_{1}$ and $g_{1}$ depend only on the known quantities $(n^{(i)},u^{(i)},T^{(i)},\phi^{(i)})$ for $i=0,1$.

By induction, let $N\geq 3$ and $1\leq k\leq N$ be integers and assume that all the systems for $(n^{(i)},u^{(i)},T^{(i)},\phi^{(i)})$ for $0\leq i\leq k-1$ have already been obtained. Then at the order $O(\varepsilon^k)$, we derive the system satisfied by $(n^{(k)},u^{(k)},T^{(k)},\phi^{(k)})$ as follows,
\begin{subequations}\label{jk}
\begin{numcases}{(\mathcal S_k)}
\mathrm{div}u^{(k)}=-\partial_t{\Delta\phi^{(k-1)}} -\sum_{i=1}^{k-1}\mathrm{div}(n^{(i)}u^{(k-i)})-u^{(0)} \cdot\nabla\Delta\phi^{(k-1)},\\
\partial_t{u^{(k)}}+{u^{(0)}}\cdot\nabla{u^{(k)}} +{u^{(k)}}\cdot\nabla{u^{(0)}}+\nabla T^{(k)}-\mu\Delta u^{(k)}-\nabla\phi^{(k)}\notag\\\ \ \ \ \ \ \ \ \ \ \ =-f_{k-1}, \ \ \ \ \ \ \ \ \\
\partial_t{T^{(k)}}+{u^{(k)}}\cdot\nabla{T^{(0)}} +{u^{(0)}}\cdot\nabla{T^{(k)}}-\frac{2\kappa}{3}\Delta T^{(k)}\notag\\\ \ \ \ \ \ \ \ \ \ \ -\frac {2\mu}3(\nabla u^{(0)}+(\nabla u^{(0)})^{\top})(\nabla u^{(k)}+(\nabla u^{(k)})^{\top})=-g_{k-1},\\
\Delta\phi^{(k-1)}=n^{(k)},
\end{numcases}
\end{subequations}
where $f_{k-1}$ and $g_{k-1}$ depends only on $(n^{(i)},u^{(i)},T^{(i)},\phi^{(i)})$ for any $0\leq i\leq k-1$ and $1\leq k\leq N$. For the solvability of $(n^{(k)},u^{(k)},T^{(k)},\phi^{(k)})$ for $1\leq k\leq N$, we have the following
\begin{theorem}\label{thm2}
Let $k\geq 1$ and $\tilde s_k\leq \tilde s-3k$ be sufficiently large integers. Then for any $0<\tau<\tau_{*}$ and any given initial data $(n^{(k)}_0,u^{(k)}_0,T^{(k)}_0,\phi^{(k)}_0)\in H^{\tilde s_k}$, the initial value problem \eqref{jk} has a unique solution $(n^{(k)},u^{(k)},T^{(k)},\phi^{(k)})\in L^{\infty}(0,\tau; H^{\tilde s_k}).$
\end{theorem}

\subsection{Main result}\label{sect-rem}
To make the above formal derivation rigorous, we consider the following expansion
\begin{subequations}\label{expan}
\begin{numcases}{}
n=1+\varepsilon^{1} n^{(1)}+\varepsilon^{2} n^{(2)}+\varepsilon^{3} n^{(3)}+\varepsilon^{4} n^{(4)}+\cdots+\varepsilon^{N} n^{(N)}+\varepsilon^{N} N_{R},\\
u=u^{(0)}+\varepsilon^{1} u^{(1)}+\varepsilon^{2} u^{(2)}+\varepsilon^{3}u^{(3)}+\varepsilon^{4} u^{(4)}+\cdots+\varepsilon^{N} u^{(N)}+\varepsilon^{N} U_{R},\\
T=T^{(0)}+\varepsilon^{1} T^{(1)}+\varepsilon^{2} T^{(2)}+\varepsilon^{3}T^{(3)}+\varepsilon^{4} T^{(4)}+\cdots+\varepsilon^{N} T^{(N)}+\varepsilon^{N} T_{R},\\
\phi=\phi^{(0)}+\varepsilon^{1}\phi^{(1)}+\varepsilon^{2}\phi^{(2)} +\varepsilon^{3}\phi^{(3)}+\varepsilon^{4}\phi^{(4)}+\cdots+\varepsilon^{N}\phi^{N}+\varepsilon^{N}\phi_{R},
\end{numcases}
\end{subequations}
where $(u^{(0)},T^{(0)},\phi^{(0)})$ satisfies \eqref{e0}, $(n^{(k)},u^{(k)},T^{(k)},\phi^{(k)})$ satisfies \eqref{jk} for $1\leq k\leq N$ and $(N_R,U_R,T_R,\phi_R)$ is the remainder.

For convenience, we use the simplified notations
\begin{equation}\label{e1}
\begin{cases}
\tilde{n}=&n^{(1)}+\varepsilon n^{(2)}+\cdots+\varepsilon^{N-1}n^{(N)},\\
\tilde{u}=&u^{(1)}+\varepsilon u^{(2)}+\cdots+\varepsilon^{N-1}u^{(N)},\\
\tilde{T}=&T^{(1)}+\varepsilon T^{(2)}+\cdots+\varepsilon^{N-1}T^{(N)},\\
\tilde{\phi}=&\phi^{(1)}+\varepsilon \phi^{(2)}+\cdots+\varepsilon^{N-1}\phi^{(N)}.
\end{cases}
\end{equation}

Now the remainder system satisfied by $(N_R,U_R,T_R,\phi_R)$ can be obtained through careful computations, and we derive the following
\begin{subequations}\label{rem1}
\begin{numcases}{}
\partial_tN_{R}+u\cdot\nabla N_{R}+n\mathrm{div}U_{R}=-\varepsilon U_{R}\cdot\nabla \tilde{n}-\varepsilon\mathrm{div}\tilde{u}N_{R}-\varepsilon\Re_{1},\label{rem1-1}\\
\partial_t{U_{R}}+u\cdot\nabla U_{R}+\frac{T\nabla N_{R}}{n}-\frac{\hbar^2}{12}\frac{\nabla\Delta N_{R}}{n}-\varepsilon^{-\frac{1}{2}}\nabla\Phi_{R}-\mu\frac{\Delta U_{R}}{n}\notag\\\ \ \ \ \ \ -(\mu+\lambda)\frac{\nabla\mathrm{div}U_{R}}{n}=F,\label{rem1-2}\\
\partial_tT_{R}+u\cdot\nabla T_{R}-\frac{2\kappa}{3}\frac{\Delta T_{R}}{n}+\frac{\hbar^2}{36}\Big(\mathrm{div}\Delta U_{R}+\varepsilon^{N}\frac{\nabla N_{R}\cdot\Delta U_{R}}{n}\Big)=G,\label{rem1-3}\\
\varepsilon^{\frac{1}{2}}\Delta\Phi_{R}=-\varepsilon\Delta\phi^{(N)}+N_{R},\label{rem1-4}
\end{numcases}
\end{subequations}
where $\Phi_{R}=\sqrt{\varepsilon}\phi_{R}$ and in \eqref{rem1-2} and \eqref{rem1-3}, we denote
\begin{equation}
\begin{split}{}
F=&-{U_{R}}\cdot\nabla(u^{(0)}+\varepsilon\tilde{u})-\nabla T_{R}-\varepsilon\frac{\nabla\tilde{n}T_{R}}{n}-\frac{\hbar^2}{12}\Big\{\varepsilon\frac{\nabla N_{R}\Delta\tilde{n}+\nabla N_{R}\cdot\nabla^{2}\tilde{n}}{n^{2}}\notag\\&-\varepsilon^{2}\frac{2(\nabla\tilde{n}\cdot\nabla N_{R})\nabla\tilde{n}+(\nabla\tilde{n}\cdot\nabla\tilde{n})\nabla N_{R}}{n^{3}}+\varepsilon\frac{\nabla\tilde{n}\Delta N_{R}+\nabla\tilde{n}\cdot\nabla^{2}N_{R}}{n^{2}}\notag\\&-\varepsilon^{N+1}\frac{2(\nabla\tilde{n}\cdot\nabla N_{R})\nabla N_{R}+(\nabla N_{R}\cdot\nabla N_{R})\nabla\tilde{n}}{n^{3}}+\varepsilon^{N}\frac{\nabla N_{R}\Delta N_{R}}{n^{2}}\notag\\&+\varepsilon^{N}\frac{\nabla N_{R}\cdot\nabla^{2}N_{R}}{n^{2}}-\varepsilon^{2N}\frac{(\nabla N_{R}\cdot\nabla N_{R})\nabla N_{R}}{n^{3}}\Big\}-\varepsilon\Re_{21}+\frac{\varepsilon\Re_{22}+\Re_{23}}{n^{3}},
\end{split}
\end{equation}
with
\begin{equation}
\begin{split}{}
G=&-U_{R}\cdot(\nabla T^{(0)}+\varepsilon\nabla\tilde{T})-\frac23(\varepsilon T_{R}\mathrm{div}\tilde{u}+T\mathrm{div}U_{R})-\frac{\hbar^2}{36}\Big\{\frac{\nabla N_{R}\cdot\Delta u^{(0)}}{n}\notag\\&-\varepsilon\frac{\nabla N_{R}\cdot\Delta \tilde{u}}{n}+\varepsilon\frac{\nabla\tilde{n}\cdot\Delta U_{R}}{n}\Big\}+\frac{2\lambda}{3}\Big\{\frac{2\varepsilon\mathrm{div}\tilde{u}\mathrm{div}U_{R}}{n}+\frac{\varepsilon^{N}(\mathrm{div}U_{R})^{2}}{n}\Big\}\notag\\&+\frac{\mu}{3}\Big\{\frac{2(\nabla u^{(0)}+(\nabla u^{(0)})^{\top}+\varepsilon\nabla\tilde{u}+\varepsilon(\nabla\tilde{u})^{\top} )\cdot(\nabla U_{R}+(\nabla U_{R})^{\top})}{n}\notag\\&+\varepsilon^{N}\frac{|\nabla U_{R}+(\nabla U_{R})^{\top}|^{2}}{n}\Big\} -\varepsilon\Re_{31}+\frac{\varepsilon\Re_{32}+\Re_{33}}{n}.
\end{split}
\end{equation}
Moreover, $\Re_{ij}$ are given by
\begin{equation}\label{rem2}
\begin{cases}
\Re_{1}= \sum_{1\leq\alpha,\beta\leq N;\alpha+\beta\geq N+1}\varepsilon^{\alpha+\beta-N-1}\mathrm{div} (n^{(\alpha)}u^{(\beta)}),\\
\Re_{21}= \sum_{1\leq\alpha,\beta\leq N;\alpha+\beta\geq N+1}\varepsilon^{\alpha+\beta-N-1}u^{(\alpha)} \cdot\nabla u^{(\beta)},\\
\Re_{22}= n^{2}\{\text{finite combination of } \nabla n^{(\alpha)}\text{ and }T^{(\beta)}\}\\\ \ \ \ \ \ \ \ +\hbar^{2}\{\text {finite combination of }n^{(\alpha)}\text{ and }\ \text{their derivatives}\}\\\ \ \ \ \ \ \ \ -\hbar^{2}n\{\text {finite combination of } n^{(\alpha)}\text{ and their derivatives up to}\\\ \ \ \ \ \ \ \ \text{second order}\}+n^{2}\hbar^{2}\{\text {finite combination of }\ n^{(\alpha)}\text{ and their derivatives}\\\ \ \ \ \ \ \ \ \text{up to}\text{ third order}\}(1\leq\alpha\leq N,0\leq\beta\leq N),\\
\Re_{23}=[\varepsilon b^{2,1}-\mu b^{2,4}-(\mu+\lambda)\varepsilon b^{2,5}]n^{2}N_{R}+\hbar^{2}[\big(-\varepsilon b^{2,2}n^{2}+\varepsilon^{2}nb^{2,3,2}\\\ \ \ \ \ \ \ \ \  -\varepsilon^{3}b^{2,3,1})N_{R} +(\varepsilon^{N+2}nb^{2,3,2} -\varepsilon^{N+3}b^{2,3,1})N_{R}^{2} -\varepsilon^{2N+3}b^{2,3,1}N_{R}^{3}],\\
\Re_{31}=\sum_{1\leq\alpha,\beta\leq N;\alpha+\beta\geq N+1}\varepsilon^{\alpha+\beta-N-1}[u^{(\alpha)} \cdot\nabla T^{(\beta)}+\mathrm{div}u^{(\alpha)}T^{(\beta)}],\\
\Re_{32}=-\frac{2\kappa}3\{\text {finite combination of }n^{(\alpha)}\text{ and }\Delta T^{(\beta)}\}+\frac{\hbar^{2}}{36}\{\text {finite combination of}\\\ \ \ \ \ \ \ \ \ n^{(\alpha)} \text {and their derivatives and }\Delta u^{(\beta)}\}+(\mu+\lambda)\{\text {finite combination of }\\\ \ \ \ \ \ \ \ \ \mathrm{div}u^{(\beta)}{ and }\ \nabla u^{(\beta)}\}(1\leq\alpha,\beta\leq N),\\
\Re_{33}=[b^{3,1}+\hbar^{2}b^{3,2} +(\mu+\lambda)\varepsilon^{2}b^{3,4}+\mu b^{3,3}]N_{R},
\end{cases}
\end{equation}
where $b^{i,j}$ are $O(1)$ terms depending only on $(n^{(k)},u^{(k)},T^{(k)},\phi^{(k)})$ for $0\leq k\leq N$ with finite terms.
The main result of this paper is stated in the following
\begin{theorem}\label{thm3}
Let $\tilde{s},\tilde{s_{k}}\geq6$, $1\leq k\leq N$, in Theorem \ref{thm1} and Theorem \ref{thm2} be sufficiently large and $(u^{(0)},T^{(0)},\phi^{(0)})\in H^{\tilde{s}}$ be a solution constructed in Theorem \ref{thm1} with initial data $(u^{(0)}_{0},T^{(0)}_{0},\phi^{(0)}_{0})\in H^{\tilde{s}}$, where $u^{(0)}_0$ is divergence free. Let $(n^{(k)},u^{(k)},T^{(k)},\phi^{(k)})\in H^{\tilde s_k}$ be a solution constructed in Theorem \ref{thm2} with initial data $(n^{(k)}_{0},u^{(k)}_{0},T^{(k)}_{0},\phi^{(k)}_{0})\in H^{\tilde s_k}$. Let $$(N_{R},U_{R},T_{R},\phi_{R})(x,0) =(N_{R0},U_{R0},T_{R0},\phi_{R0})\in H^{5}$$ and assume the initial data $(n,u,T,\phi)(x,0)=(n_{0},u_{0},T_{0},\phi_{0})$ of \eqref{equ2} has the expansion
\begin{equation}
\begin{split}
(n_{0},u_{0},T_{0},\phi_{0}) = &(1,u^{(0)}_{0},T^{(0)}_{0},\phi^{(0)}_{0}) +\sum_{j=1}^{N}\varepsilon^{j}(n^{(j)}_{0}, u^{(j)}_{0},T^{(j)}_{0},\phi^{(j)}_{0}) \notag\\&+\varepsilon^{N}(N_{R0},U_{R0},T_{R0},\phi_{R0}).
\end{split}
\end{equation}
Then there exists $\varepsilon_{0}=\varepsilon_{0}(\tau)$ such that if $0<\varepsilon<\varepsilon_{0}$ there is a maximal time interval $[0,\tau_{\varepsilon}]$ with $\liminf_{\varepsilon\rightarrow 0}\tau_{\varepsilon}\geq\tau$ such that \eqref{equ2} has a unique solution $(n,u,T,\phi)$ satisfying
\begin{equation}
\begin{split}
&\big(n-(1+\varepsilon\tilde{n}), u-(u^{{(0)}}+\varepsilon\tilde{u}), T-(T^{{(0)}}+\varepsilon\tilde{T}), \nabla\phi-(\nabla\phi^{{(0)}}+\varepsilon\nabla\tilde{\phi})\big)\notag\\&\ \ \  \ \ \ \ \ \ \ \ \ \ \ \in C([0,\tau],H^{5}\times H^{4}\times H^{3}\times H^{4}).
\end{split}
\end{equation}
Moreover, there exists $C>0$ independent of $\varepsilon>0$ but may possibly depend on $\tau$, such that
$$\sup_{t\in[0,\tau]}\|\big(n-(1+\varepsilon\tilde{n}), u-(u^{{(0)}}+\varepsilon\tilde{u}), T-(T^{{(0)}}+\varepsilon\tilde{T})\big)\|_{H^{3}} \leq C\varepsilon^{N},$$ where $(\tilde n,\tilde u,\tilde T,\tilde \phi)$ is given in \eqref{e1}.
\end{theorem}

In the next two sections, we will give some uniform estimates for $(N_R,U_R,T_R,\Phi_R)$. Finally in Section \ref{proof}, we prove Theorem \ref{thm3}. Throughout this paper, we use to $a\lesssim b$ to stand for $a\leq Cb$ for some constant $C>0$. Let $[f,g]=fg-gf$ to denote the commutator of $f$ and $g$.

\section{Basic estimates}
To give uniform estimates for \eqref{rem1}, we first introduce the norm
\begin{equation}\label{def-A}
\begin{split}
|\!|\!|(N_{R},U_R,T_{R},\nabla\Phi_{R})|\!|\!|_{3}^{2}=&\|(N_{R},U_R,T_{R},\nabla\Phi_{R})\|_{H^{3}}^{2}\\
&+\|(\hbar\nabla N_{R},\hbar\nabla U_R,\hbar\mathrm{div}U_R,\hbar\Delta\Phi_{R})\|_{H^{3}}^{2}+\|\hbar^{2}\Delta N_{R}\|_{H^{3}}^{2}.
\end{split}
\end{equation}
Let $\tilde C$ be a constant to be determined later, which is assumed to be independent of $\varepsilon$ and much larger than the bound $|\!|\!|(U_R,N_{R},T_{R},\nabla\Phi_R)(0)|\!|\!|_{3}^2$ of the initial data. Next, it is classical that there exists $\tau_{\varepsilon}>0$ such that on $[0,\tau_{\varepsilon}]$,
\begin{equation}\label{prior}
\begin{split}
|\!|\!|(N_{R},U_R,T_{R},\nabla\Phi_{R})|\!|\!|_{3}^{2}\leq \tilde C.
\end{split}
\end{equation}
Hence, we can assume that $n$ is bounded from above and below
\begin{equation}\label{f1}
\begin{split}
1/2<n<3/2,
\end{split}
\end{equation}
 when $\varepsilon$ is sufficiently small.

We will show that for any given $\tau>0$, there is some $\varepsilon_0>0$ such that the existence time $\tau_{\varepsilon}>\tau$ for any  $0<\varepsilon<\varepsilon_0$. To prove the theorem \ref{thm3}, we need to derive the uniform estimate for the remainder system \eqref{rem1}. To this end, we first give some estimates in Lemma \ref{L1}, Lemma \ref{L2} and Lemma \ref{L3} in this section.
\begin{lemma}\label{L1}
Let $0\leq k\leq 3$ be an integer, $(N_{R},U_{R},T_{R},\Phi_{R})$ be a solution to \eqref{rem1}, and $\alpha$ be a multi-index with $|\alpha|=k$, then we obtain
\begin{equation}
\begin{split}
\|\partial_tN_{R}\|_{L^{\infty}}^2\lesssim & 1+\|(U_{R},N_{R})\|_{H^{3}}^2,\\
\|\partial_tN_{R}\|_{H^{k}}^2\lesssim & (1+\varepsilon^{N}\|(N_{R},U_{R})\|_{H^{3}})\|(\mathrm{div}U_{R},\nabla N_{R},U_{R},N_{R})\|_{H^{3}}+\varepsilon,
\end{split}
\end{equation}
and
\begin{equation}
\begin{split}
\|\partial_t(\frac{1}{n})\|_{L^{\infty}}^2\lesssim & \varepsilon+\varepsilon^{N}\|(U_{R},N_{R})\|_{H^{3}}^2.
\end{split}
\end{equation}
\end{lemma}
\begin{proof}
Notice that applying \eqref{rem1-1}, \eqref{f1}, and Sobolev embedding, together with \eqref{R1} and \eqref{r2} in the Appendix, we can obtain
\begin{equation*}
\begin{split}
\|\partial_tN_{R}\|_{L^{\infty}}
\lesssim &(1+\varepsilon^{N}\|U_{R}\|_{L^{\infty}})\|(N_{R},U_{R})\|_{H^{3}}+\varepsilon\lesssim 1+\|(U_{R},N_{R})\|_{H^{3}}^2,
\end{split}
\end{equation*}
and
\begin{equation*}
\begin{split}
\|\partial_tN_{R}\|_{H^{k}}
\lesssim &(1+\varepsilon^{N}\|(N_{R},U_{R})\|_{H^{3}})\|(\mathrm{div}U_{R},\nabla N_{R},U_{R},N_{R})\|_{H^{3}}+\varepsilon.
\end{split}
\end{equation*}
Moreover, it is easy to have
\begin{equation*}
\begin{split}
\|\partial_t\frac{1}{n}\|_{L^{\infty}}= &\|\frac{\partial_tn}{n^{2}}\|_{L^{\infty}} \lesssim \varepsilon+\varepsilon^{N}\|\partial_tN_{R}\|_{L^{\infty}} \lesssim\varepsilon+\varepsilon^{N}\|(U_{R},N_{R})\|_{H^{3}}^2.
\end{split}
\end{equation*}
\end{proof}
\begin{lemma}\label{L2}
Under the same conditions in Lemma \ref{L1}, we have the following estimate,
\begin{equation}
\begin{split}
\|\partial^{\alpha}(\frac{1}{n})\|_{L^\infty}\lesssim &1+\varepsilon^{|\alpha|N}\|N_{R}\|_{H^{2+|\alpha|}}^{|\alpha|},\\
\|\partial^{\alpha}(\frac{1}{n})\|_{L^2}\lesssim &1+\varepsilon^{|\alpha|N}\|N_{R}\|_{H^{|\alpha|}}^{|\alpha|},
\end{split}
\end{equation}
\begin{equation}
\begin{split}
\|\partial^{\alpha}(\frac{u}{n})\|_{L^\infty}\lesssim &1+\varepsilon^{(1+|\alpha|)N}\|(N_{R},U_{R})\|_{H^{2+|\alpha|}}^{1+|\alpha|},\\
\|\partial^{\alpha}(\frac{u}{n})\|_{L^2}\lesssim &1+\varepsilon^{(1+|\alpha|)N}\|(N_{R},U_{R})\|_{H^{|\alpha|}}^{1+|\alpha|},
\end{split}
\end{equation}
and
\begin{equation}
\begin{split}
\|\partial^{\alpha}(\frac{T}{n})\|_{L^\infty}\lesssim &1+\varepsilon^{(1+|\alpha|)N}\|(N_{R},T_{R})\|_{H^{2+|\alpha|}}^{1+|\alpha|},\\
\|\partial^{\alpha}(\frac{T}{n})\|_{L^2}\lesssim &1+\varepsilon^{(1+|\alpha|)N}\|(N_{R},T_{R})\|_{H^{|\alpha|}}^{1+|\alpha|}.
\end{split}
\end{equation}
\end{lemma}
\begin{proof}
For $|\alpha|=1$, it is easy to prove, so we take $|\alpha|=2$ for example. By computation
\begin{equation*}
\begin{split}
\partial^{2}(\frac{1}{n})= &\frac{2(\varepsilon\partial\tilde{n}+\varepsilon^{N}\partial N_{R})^{2}}{n^{3}}-\frac{\varepsilon\partial^{2}\tilde{n}+\varepsilon^{N}\partial^{2} N_{R}}{n^{2}},
\end{split}
\end{equation*}
then taking $L^{2}$-norm, using H\"older's inequality and Sobolev embedding $H^{1}\hookrightarrow L^{3},L^{6}$, we can obtain
\begin{equation*}
\begin{split}
\|\partial^{2}(\frac{1}{n})\|_{L^{2}}\lesssim &4\varepsilon^{N+1}\|\partial\tilde{n}\partial N_{R}\|_{L^{2}}+2\varepsilon^{2}\|(\partial\tilde{n})^{2}\|_{L^{2}}+2\varepsilon^{2N}\|(\partial N_{R})^{2}\|_{L^{2}}\\&+\varepsilon+\varepsilon^{N}\|\partial^{2}N_{R}\|_{L^{2}}
\lesssim 1+\varepsilon^{2N}\|N_{R}\|_{H^{2}}^{2}.
\end{split}
\end{equation*}
Similarly, taking $L^{\infty}$-norm yields
\begin{equation*}
\begin{split}
\|\partial^{2}(\frac{1}{n})\|_{L^{\infty}}\lesssim &1+\varepsilon^{2N}\|N_{R}\|_{H^{4}}^{2}.
\end{split}
\end{equation*}
Similarly, by computation, we have
\begin{equation*}
\begin{split}
\partial^{2}(\frac{u}{n})=&\partial^{2}(\frac{1}{n})u+2\partial(\frac{1}{n})\partial u+\frac{\partial^{2}u}{n}.
\end{split}
\end{equation*}
Due to Sobolev embedding $H^{2}\hookrightarrow L^{\infty}$, Young's inequality and the above results, we can verify
\begin{equation*}
\begin{split}
\|\partial^{2}(\frac{u}{n})\|_{L^{\infty}}\lesssim &(1+\varepsilon^{2N}\|N_{R}\|_{H^{4}}^{2})(1+\varepsilon^{N}\|U_{R}\|_{H^{2}})++(1+\varepsilon^{N}\|U_{R}\|_{H^{4}})\notag\\&
(1+\varepsilon^{N}\|N_{R}\|_{H^{3}})(1+\varepsilon^{N}\|U_{R}\|_{H^{3}})\lesssim 1+\varepsilon^{3N}\|(N_{R},U_{R})\|_{H^{4}}^{3}.
\end{split}
\end{equation*}
Again taking $L^{2}$-norm, using H\"older's inequality and Sobolev embedding $H^{1}\hookrightarrow L^{3},L^{6}$, we know
\begin{equation*}
\begin{split}
\|\partial^{2}(\frac{u}{n})\|_{L^{2}}\lesssim &\|u\|_{H^{2}}\|\frac 1n\|_{H^{2}}
\lesssim 1+\varepsilon^{3N}\|(N_{R},U_{R})\|_{H^{2}}^{3}.
\end{split}
\end{equation*}
The case of $|\alpha|\neq2$ can be proved similarly, thus we omit them.
\end{proof}
\begin{lemma}\label{L3}
Let $\alpha$ be a multi-index with $0\leq|\alpha|\leq k$ and $k$ be an integer, $f\in \mathbb{S},$ the Schwartz class, then we obtain a prior bound
\begin{equation}
\begin{split}
\|\nabla^{2}f\|_{H^{k}}\lesssim &\|\Delta f\|_{H^{k}}.
\end{split}
\end{equation}
\end{lemma}
\begin{proof}
Thanks to Plancherel's theorem, we derive
\begin{equation}
\begin{split}
\|\partial^{\alpha}\partial_{i}\partial_{j}f\|_{L^{2}}=&\|(2\pi i\xi)^{\alpha}(-4\pi^{2}\xi_{i}\xi_{j})\hat{f}\|_{L^{2}}=\|(2\pi i\xi)^{\alpha}(-\frac{i\xi_{i}i\xi_{j}}{|\xi|^{2}})(-4\pi^{2}|\xi|^{2})\hat{f}\|_{L^{2}}\\
= &\|\partial^{\alpha}R_{i}R_{j}(-\Delta)f\|_{L^{2}},
\end{split}
\end{equation}
where we use the Riesz operator $R_{j},$ $\widehat{(R_{j}f)}=\frac{i\xi_{j}}{|\xi|}\hat{f},$ and the conclusion in \cite{Stein1970} that $R_{i}R_{j}$ is bounded from $L^{p}$ to $L^{p}$ with $1< p<\infty.$ Then summing all multi-index $\alpha$ with $0\leq|\alpha|\leq k$, we complete the proof.
\end{proof}

\section{Uniform energy estimates}
In this section, we will prove the following
\begin{proposition}\label{prop1}
Let $(N_{R},U_{R},T_{R},\Phi_{R})$ be a solution to \eqref{rem1}, then there holds
\begin{equation}\label{prop}
\begin{split}
&|\!|\!|(N_{R},U_R,T_{R},\nabla\Phi_{R})|\!|\!|_{3}^{2}+\int_{0}^{t}\Big\{\kappa\|\nabla T_{R}\|_{H^{3}}^2+\mu\|\nabla U_{R}\|_{H^{3}}^2+(\mu+\lambda)\|\mathrm{div}U_{R}\|_{H^{3}}^2\\&\ \ \ \ +\hbar^{2}\mu\|\Delta U_{R}\|_{H^{3}}^2 +\hbar^{2}(\mu+\lambda)\|\nabla\mathrm{div}U_{R}\|_{H^{3}}^2\Big\}\\&\lesssim |\!|\!|(N_{R},U_R,T_{R},\nabla\Phi_{R})(0)|\!|\!|_{3}^{2}\\&\ \ \ \ + \int_{0}^{t}\Big\{(1+\varepsilon^{5N}\|(N_{R},U_{R},T_{R})\|_{H^{3}}^{10})|\!|\!|(N_{R},U_R,T_{R},\nabla\Phi_{R})|\!|\!|_{3}^{2}+\varepsilon\Big\},
\end{split}
\end{equation}
where $\delta= \min\{\frac 1{32},\frac 1{32\mu},\frac 1{32(\mu+\lambda)}\}$ and $\hbar\ll 1$.
\end{proposition}

This proposition is proved as a direct sequence of the following Lemmas. Moreover, combining Proposition \ref{prop1} with the standard continuous induction method, we can obtain for any given $\tau>0$, there is some $\varepsilon_0>0$ such that the existence time $\tau_{\varepsilon}>\tau$ for any  $0<\varepsilon<\varepsilon_0$.
\begin{lemma}\label{G1}
Let $0\leq k\leq 3$ be an integer, $(N_{R},U_{R},T_{R},\Phi_{R})$ be a solution to \eqref{rem1}, and $\alpha$ be a multi-index with $|\alpha|=k$, then we can have
\begin{equation}\label{g1}
\begin{split}
\frac{d}{dt}\|N_{R}\|_{H^{\alpha}}^2\lesssim &\frac{\mu+\lambda}{32}\|\mathrm{div}U_{R}\|_{H^{k}}^2+(1+\varepsilon^{N}\|(N_{R},U_{R})\|_{H^{3}})\|(N_{R},U_{R})\|_{H^{3}}^2+\varepsilon.
\end{split}
\end{equation}
\end{lemma}
\begin{proof}
Applying the operator $\partial^{\alpha}$ to \eqref{rem1-1} and taking inner product with $\partial^{\alpha}N_{R}$,  it holds
\begin{equation}\label{d1}
\begin{split}
&\frac{1}{2}\frac{d}{dt}\|\partial^{\alpha}N_{R}\|^2\\= &-\int \partial^{\alpha}(u\cdot\nabla N_{R})\partial^{\alpha}N_{R}-\int\partial^{\alpha}(n\mathrm{div}U_{R})\partial^{\alpha}N_{R}-\varepsilon\int\partial^{\alpha}(U_{R}\cdot\nabla\tilde{n})\partial^{\alpha}N_{R}\notag\\&-\varepsilon\int\partial^{\alpha}(N_{R}\mathrm{div}\tilde{u})\partial^{\alpha}N_{R}-\varepsilon\int\partial^{\alpha}\Re_{1}\partial^{\alpha}N_{R}.
\end{split}
\end{equation}
We shall estimate the first two terms of the right hand side. Combining the fact $\mathrm{div}u^{(0)}=0$ with integration by parts, Sobolev embedding and the commutator estimates Remark \ref{R1} in the Appendix, we obtain
\begin{equation*}
\begin{split}
&-\int \partial^{\alpha}(u\cdot\nabla N_{R})\partial^{\alpha}N_{R}-\int\partial^{\alpha}(n\mathrm{div}U_{R})\partial^{\alpha}N_{R}\\=&-\int[\partial^{\alpha},u]\nabla N_{R}\partial^{\alpha}N_{R}+\frac{1}{2}\int \mathrm{div}u|\partial^{\alpha}N_{R}|^{2}-\int[\partial^{\alpha},n]\mathrm{div} U_{R}\partial^{\alpha}N_{R}\notag\\&-\int n\partial^{\alpha}(\mathrm{div} U_{R})\partial^{\alpha}N_{R}\\
\lesssim &(1+\varepsilon^{N}\|(N_{R},U_{R})\|_{H^{3}})\|(N_{R},U_{R})\|_{H^{3}}^2+(\mu+\lambda)\delta\|\mathrm{div} U_{R}\|_{H^{k}}^{2}+\frac{1}{(\mu+\lambda)\delta}\|N_{R}\|_{H^{k}}^{2},
\end{split}
\end{equation*}
for any sufficiently small positive constant $\delta$.
It is obvious that the last three terms can be bounded by $\varepsilon+\|(N_{R},U_{R})\|_{H^{3}}^{2}.$
The proof is complete by taking $\delta=\frac 1{32}$.
\end{proof}
\begin{lemma}\label{G2}
Under the assumptions in Lemma \ref{G1}, we obtain
\begin{equation}\label{g2}
\begin{split}
&\frac{d}{dt}\|\partial^{\alpha}(U_{R},\hbar\nabla N_{R},\nabla\Phi_{R})\|_{L^{2}}^2+\mu\|\partial^{\alpha}\nabla U_{R}\|_{L^{2}}^2+(\mu+\lambda)\|\partial^{\alpha}\mathrm{div}U_{R}\|_{L^{2}}^2\\&\lesssim(1+\varepsilon^{3N}\|(U_{R},T_{R},N_{R},\hbar\nabla N_{R})\|_{H^{3}}^{6})\|(U_{R},N_{R},T_{R},\nabla\Phi_{R},\hbar\nabla N_{R},\hbar\nabla U_{R},\hbar\mathrm{div}U_{R},\\&\ \ \ \ \hbar^{2}\Delta N_{R} )\|_{H^{3}}^2+\frac \kappa{32}\|\nabla T_{R}\|_{H^{k}}^2+\varepsilon.
\end{split}
\end{equation}
\end{lemma}
\begin{proof}
Applying the operator $\partial^{\alpha}$ to \eqref{rem1-2} and taking inner product with $\partial^{\alpha}U_{R}$, we derive
\begin{equation}\label{g2}
\begin{split}
\frac{1}{2}\frac{d}{dt}\|\partial^{\alpha}U_{R}\|_{L^{2}}^2=& \int\partial^{\alpha}F\partial^{\alpha}U_{R}-\int\partial^{\alpha}(u\cdot\nabla U_{R})\cdot \partial^{\alpha}U_{R}-\int\partial^{\alpha}(\frac{T\nabla N_{R}}{n})\cdot\partial^{\alpha}U_{R}\\&
+\frac{\hbar^{2}}{12}\int\partial^{\alpha}(\frac{\nabla\Delta N_{R}}{n})\cdot\partial^{\alpha}U_{R}+\mu\int\partial^{\alpha}(\frac{\Delta U_{R}}{n})\cdot\partial^{\alpha}U_{R}
\\&+(\mu+\lambda)\int\partial^{\alpha}(\frac{\nabla\mathrm{div}U_{R}}{n})\cdot\partial^{\alpha}U_{R}+\frac{1}{\varepsilon^{\frac{1}{2}}}\int\partial^{\alpha}\nabla\Phi_{R}\cdot\partial^{\alpha}U_{R}\\
&= \sum_{i=1}^{7}R_{2,i}.
\end{split}
\end{equation}
The first term $R_{2,1}$ can be bounded by Young's inequality, Lemma \ref{L2}, Lemma \ref{L3} and \eqref{r1}$\sim$\eqref{r5} in the Appendix that
\begin{equation*}
\begin{split}
R_{2,1}\lesssim & \kappa\delta\|\nabla T_{R}\|_{H^{k}}^{2}\notag\\&+(1+\varepsilon^{5N}\|(N_{R},\hbar\nabla N_{R})\|_{H^{3}}^{5})\|(\hbar^{2}\Delta N_{R},\hbar\nabla N_{R},U_{R},N_{R},T_{R})\|_{H^{3}}^{2}+\varepsilon,
\end{split}
\end{equation*}
for any sufficiently small positive constant $\delta$. Here, we only estimate the two terms in $R_{2,1}$,
\begin{equation*}
\begin{split}
&-\frac{\hbar^{2}}{12}\int\partial^{\alpha}\Big(\varepsilon^{N}\frac{\nabla N_{R}\Delta N_{R}+\nabla N_{R}\cdot\nabla^{2}N_{R}}{n^{2}}-\varepsilon^{2N}\frac{(\nabla N_{R}\cdot\nabla N_{R})\nabla N_{R}}{n^{3}}\Big)\partial^{\alpha}U_{R}\\
\lesssim &\hbar^{2}\|U_{R}\|_{H^{3}}\Big\{\varepsilon^{N}(\|\frac{1}{n^{2}}\|_{L^{\infty}}\|\nabla N_{R}\|_{H^{k}}+\|\frac{1}{n^{2}}\|_{H^{3}}\|\nabla N_{R}\|_{L^{\infty}})\|(\Delta N_{R},\nabla^{2}N_{R})\|_{L^{\infty}}\notag\\&+\varepsilon^{N}\|\frac{\nabla N_{R}}{n^{2}}\|_{L^{\infty}}\|(\Delta N_{R},\nabla^{2}N_{R})\|_{H^{k}}+\varepsilon^{2N}\|(\frac{1}{n})^{3}\|_{H^{3}}\|(\nabla N_{R})^{3}\|_{H^{3}}\Big\}\\
\lesssim &(1+\varepsilon^{5N}\|(N_{R},\hbar\nabla N_{R})\|_{H^{3}}^{5})\|(\hbar^{2}\Delta N_{R},\hbar\nabla N_{R},U_{R})\|_{H^{3}}^{2},
\end{split}
\end{equation*}
thanks to Lemma \ref{L3} and \eqref{r2}$\backsim$\eqref{r5} in the Appendix.

For the second term, based on integration by parts and the commutator estimates Remark \ref{R1} in the Appendix, we deduce
\begin{equation*}
\begin{split}
R_{2,2}= &-\int[\partial^{\alpha},u]\nabla U_{R}\partial^{\alpha}U_{R}-\int u\cdot\partial^{\alpha}\nabla U_{R}\partial^{\alpha}U_{R}\lesssim (1+\varepsilon^{N}\|U_{R}\|_{H^{3}})\|U_{R}\|_{H^{3}}^{2}.
\end{split}
\end{equation*}
Now, we will estimate the term $R_{2,3}$ since there is $\nabla N_{R}$ and we can not close our estimates with this. Owing to Lemma \ref{L2}, integration by parts, Young's inequality, and the commutator estimates Remark \ref{R1} in the Appendix, we have
\begin{equation*}
\begin{split}
R_{2,3}= &-\int\partial^{\alpha}(\frac{T\nabla N_{R}}{n})\cdot\partial^{\alpha}U_{R}= -\int[\partial^{\alpha},\frac{T}{n}]\nabla N_{R}\cdot\partial^{\alpha}U_{R}-\int\frac{T}{n}\partial^{\alpha}\nabla N_{R}\cdot\partial^{\alpha}U_{R}\\
= &-\int[\partial^{\alpha},\frac{T}{n}]\nabla N_{R}\cdot\partial^{\alpha}U_{R}+\int\nabla(\frac{T}{n})\partial^{\alpha}N_{R}\partial^{\alpha}U_{R}+\int\frac{T}{n}\partial^{\alpha}N_{R}\partial^{\alpha}\mathrm{div}U_{R}\\
\lesssim &(\mu+\lambda)\delta\|\mathrm{div}U_{R}\|_{H^{k}}^{2}+(1+\varepsilon^{4N}\|(N_{R},T_{R})\|_{H^{3}}^{4})\|(N_{R},U_{R})\|_{H^{3}}^{2},
\end{split}
\end{equation*}
for any positive constant $\delta$.

The estimate for $R_{2,4}$ in \eqref{g2} requires much efforts since it involves higher order terms. From \eqref{rem1-1} and the commutator, it is easy to compute
$$n\mathrm{div}U_{R}=-(\partial_{t}N_{R}+u\cdot\nabla N_{R}+\varepsilon U_{R}\cdot\nabla\tilde{n}+\varepsilon N_{R}\mathrm{div}\tilde{u}+\varepsilon\Re_{1}),$$and
$$n\partial^{\alpha}\mathrm{div}U_{R}=\partial^{\alpha}(n\mathrm{div}U_{R})-[\partial^{\alpha},n]\mathrm{div}U_{R}.$$
That is to say, we can insert the above results into $R_{2,4}$ to decompose
\begin{equation*}
\begin{split}
R_{2,4}= &\frac{\hbar^{2}}{12}\int\partial^{\alpha}(\frac{\nabla\Delta N_{R}}{n})\cdot\partial^{\alpha}U_{R}\\
= &\frac{\hbar^{2}}{12}\int[\partial^{\alpha},\frac{1}{n}]\nabla\Delta N_{R}\cdot\partial^{\alpha}U_{R}+\frac{\hbar^{2}}{12}\int\frac{\partial^{\alpha}(\nabla\Delta N_{R})\cdot\partial^{\alpha}U_{R}}{n}\\
= &\frac{\hbar^{2}}{12}\int[\partial^{\alpha},\frac{1}{n}]\nabla\Delta N_{R}\cdot\partial^{\alpha}U_{R}+\frac{\hbar^{2}}{12}\int\frac{\partial^{\alpha}(\Delta N_{R})\nabla n\cdot\partial^{\alpha}U_{R}}{n^{2}}\notag\\&-\frac{\hbar^{2}}{12}\int\frac{\partial^{\alpha}(\Delta N_{R})\partial^{\alpha}\mathrm{div}U_{R}}{n}\\
= &-\frac{\hbar^{2}}{24}\frac{d}{dt}\int\frac{|\partial^{\alpha}\nabla N_{R}|^{2}}{n^{2}}+\frac{\hbar^{2}}{12}\int\frac{\partial^{\alpha}\Delta N_{R}\partial^{\alpha}(u\cdot\nabla N_{R})}{n^{2}}+R_{2,4,1},
\end{split}
\end{equation*}
where $R_{2,4,1}$ is denoted by
\begin{equation*}
\begin{split}
&R_{2,4,1}\\
= &\frac{\hbar^{2}}{12}\int[\partial^{\alpha},\frac{1}{n}]\nabla\Delta N_{R}\cdot\partial^{\alpha}U_{R}+\frac{\hbar^{2}}{12}\int\frac{\partial^{\alpha}(\Delta N_{R})\nabla n\cdot\partial^{\alpha}U_{R}}{n^{2}}-\frac{\hbar^{2}}{12}\int\frac{|\partial^{\alpha}\nabla N_{R}|^{2}}{n^{3}}\partial_{t}n\notag\\&+\frac{\hbar^{2}}{6}\int\frac{\nabla n\cdot\partial^{\alpha}\nabla N_{R}\partial^{\alpha}\partial_{t}N_{R}}{n^{3}}+\frac{\hbar^{2}}{12}\varepsilon\int\frac{\partial^{\alpha}\Delta N_{R}\partial^{\alpha}(U_{R}\cdot\nabla\tilde{n}+N_{R}\mathrm{div}\tilde{u}+\Re_{1})}{n^{2}}\notag\\&
+\frac{\hbar^{2}}{12}\int\frac{\partial^{\alpha}\Delta N_{R}[\partial^{\alpha},n]\mathrm{div}U_{R}}{n^{2}}.
\end{split}
\end{equation*}
Moreover, by the commutator estimates, integration by parts, H\"older inequality, Sobolev embedding $H^{2}\hookrightarrow L^{\infty}$ and Lemma \ref{L1}, Lemma \ref{L2}, it follows that
\begin{equation*}
\begin{split}
R_{2,4,1}
\lesssim &(1+\varepsilon^{\frac 32N}\|(U_{R},N_{R})\|_{H^{3}}^{3})\|(\hbar\nabla N_{R},\hbar\mathrm{div}U_{R},\hbar^{2}\Delta N_{R},N_{R},U_{R})\|_{H^{3}}^{2}+\varepsilon.
\end{split}
\end{equation*}
Thanks to Lemma \ref{L2}, integration by parts and the commutator estimates Remark \ref{R1} in the Appendix, the rest one in $R_{2,4}$ can be bounded by
\begin{equation*}
\begin{split}
&\frac{\hbar^{2}}{12}\int\frac{\partial^{\alpha}\Delta N_{R}\partial^{\alpha}(u\cdot\nabla N_{R})}{n^{2}}\\
= &\frac{\hbar^{2}}{6}\int\frac{\nabla n\cdot\partial^{\alpha}\nabla N_{R}\partial^{\alpha}(u\cdot\nabla N_{R})}{n^{3}}-\frac{\hbar^{2}}{12}\int\frac{\partial^{\alpha}\nabla N_{R}[\partial^{\alpha}\nabla,u]\cdot\nabla N_{R}}{n^{2}}\notag\\&+\frac{\hbar^{2}}{24}\int\mathrm{div}(\frac{u}{n^{2}})|\partial^{\alpha}\nabla N_{R}|^{2}\\
\lesssim &\hbar^{2}\|\nabla n\|_{L^{\infty}}\|\nabla N_{R}\|_{H^{k}}(\|u\|_{L^{\infty}}\|\nabla N_{R}\|_{H^{k}}+\|u\|_{H^{k}}\|\nabla N_{R}\|_{L^{\infty}})\notag\\&+\hbar^{2}\|\nabla N_{R}\|_{H^{k}}\big(\|\nabla u\|_{L^{\infty}}\|\nabla N_{R}\|_{H^{k}}+\|\nabla u\|_{H^{k}}\|\nabla N_{R}\|_{L^{\infty}}\big)
\notag\\&+\hbar^{2}\|\mathrm{div}(\frac{u}{n^{2}})\|_{L^{\infty}}\|\nabla N_{R}\|_{H^{k}}^{2}\\
\lesssim &(1+\varepsilon^{2N}\|(N_{R},U_{R})\|_{H^{3}}^{2})\|(N_{R},\hbar\nabla U_{R},\hbar\nabla N_{R})\|_{H^{3}}^{2}.
\end{split}
\end{equation*}
Therefore, combining all estimates for $R_{2,4}$, we can obtain
\begin{equation*}
\begin{split}
R_{2,4}\lesssim &-\frac{\hbar^{2}}{24}\frac{d}{dt}\int\frac{|\partial^{\alpha}\nabla N_{R}|^{2}}{n^{2}}+\varepsilon\notag\\&+(1+\varepsilon^{\frac 32N}\|(U_{R},N_{R})\|_{H^{3}}^{3})\|(N_{R},U_{R},\hbar\nabla N_{R},\hbar\nabla U_{R},\hbar\mathrm{div}U_{R},\hbar^{2}\Delta N_{R})\|_{H^{3}}^2.
\end{split}
\end{equation*}
Next, we estimate the terms $R_{2,5}$. It is treated by integration by parts, Young's inequality and Remark \ref{R1} in the Appendix that
\begin{equation*}
\begin{split}
R_{2,5}= &\mu\int\partial^{\alpha}(\frac{\Delta U_{R}}{n})\cdot\partial^{\alpha}U_{R}= \mu\int[\partial^{\alpha},\frac 1n]\Delta U_{R}\cdot\partial^{\alpha}U_{R}+\mu\int\frac{\partial^{\alpha}\Delta U_{R}\cdot\partial^{\alpha}U_{R}}{n}\\
= &\mu\int[\partial^{\alpha},\frac 1n]\Delta U_{R}\cdot\partial^{\alpha}U_{R}-\mu\int\frac{|\partial^{\alpha}\nabla U_{R}|^{2}}{n}+\mu\int\frac{\nabla n\cdot\partial^{\alpha}\nabla U_{R}\cdot\partial^{\alpha}U_{R}}{n^{2}}\\
\lesssim &-\mu\int\frac{|\partial^{\alpha}\nabla U_{R}|^{2}}{n}+\mu\delta\|\nabla U_{R}\|_{H^{3}}^{2}+(1+\varepsilon^{6N}\|N_{R}\|_{H^{3}}^{6})\|U_{R}\|_{H^{k}}^{2},
\end{split}
\end{equation*}
for sufficiently small positive constant $\delta$.
Similar to $R_{2,5}$, we obtain
\begin{equation*}
\begin{split}
R_{2,6}\lesssim &-(\mu+\lambda)\int\frac{|\partial^{\alpha}\mathrm{div}U_{R}|^{2}}{n}+(\mu+\lambda)\delta\|\mathrm{div} U_{R}\|_{H^{3}}^{2}+(1+\varepsilon^{6N}\|N_{R}\|_{H^{3}}^{6})\|U_{R}\|_{H^{k}}^{2}.
\end{split}
\end{equation*}
Now we will estimate the last term. From \eqref{rem1-1}, we can compute the following equality
$$\mathrm{div}U_{R}=-\Big(\partial_{t}N_{R}+\mathrm{div}(\varepsilon\tilde{n}U_{R}+uN_{R}+\varepsilon\Re_{1})\Big).$$
Plugging this above equality into $R_{2,7}$, together with \eqref{rem1-4}, we obtain
\begin{equation*}
\begin{split}
&\frac{1}{\varepsilon^{\frac{1}{2}}}\int\partial^{\alpha}\nabla\Phi_{R}\cdot\partial^{\alpha}U_{R} =-\frac{1}{\varepsilon^{\frac{1}{2}}}\int\partial^{\alpha}\Phi_{R}\partial^{\alpha}\mathrm{div}U_{R}\\
= &\int\partial^{\alpha}\Phi_{R}\partial^{\alpha}(\Delta\partial_{t}\Phi_{R}+\varepsilon^{\frac{1}{2}}\Delta\partial_{t}\phi^{(N)})-
\frac{1}{\varepsilon^{\frac{1}{2}}}\int\partial^{\alpha}\nabla\Phi_{R}\cdot\partial^{\alpha}(\varepsilon\tilde{n}U_{R}+uN_{R}+\varepsilon\Re_{1})\\
= &-\frac{d}{dt}\frac 12\|\partial^{\alpha}\nabla\Phi_{R}\|_{L^{2}}^{2}-\varepsilon^{\frac{1}{2}}\int\partial^{\alpha}\nabla\Phi_{R}\cdot\nabla\partial_{t}\phi^{(N)}-
\varepsilon^{\frac{1}{2}}\int\partial^{\alpha}\nabla\Phi_{R}\cdot\partial^{\alpha}(\tilde{n}U_{R}+\Re_{1})\notag\\&-
\int\partial^{\alpha}\nabla\Phi_{R}\cdot\partial^{\alpha}\big(u(\Delta\Phi_{R}+\varepsilon^{\frac{1}{2}}\Delta\phi^{(N)})\big)\\
= &-\frac{d}{dt}\|\partial^{\alpha}\nabla\Phi_{R}\|_{L^{2}}^{2}-\varepsilon^{\frac{1}{2}}\int\partial^{\alpha}\nabla\Phi_{R}\cdot\nabla\partial_{t}\phi^{(N)}
-\varepsilon^{\frac{1}{2}}\int\partial^{\alpha}\nabla\Phi_{R}\cdot\partial^{\alpha}(\tilde{n}U_{R}+\Re_{1})\notag\\&-
\int\partial^{\alpha}\nabla\Phi_{R}\cdot[\partial^{\alpha},u]\Delta\Phi_{R}-\int\partial^{\alpha}\nabla\Phi_{R}\cdot (u\partial^{\alpha}\Delta\Phi_{R})-\varepsilon^{\frac{1}{2}}\int\partial^{\alpha}\nabla\Phi_{R}\cdot\partial^{\alpha}(u\Delta\phi^{(N)})\\
= &-\frac{d}{dt}\|\partial^{\alpha}\nabla\Phi_{R}\|_{L^{2}}^{2}-\varepsilon^{\frac{1}{2}}\int\partial^{\alpha}\nabla\Phi_{R}\cdot\nabla\partial_{t}\phi^{(N)}
-\varepsilon^{\frac{1}{2}}\int\partial^{\alpha}\nabla\Phi_{R}\cdot\partial^{\alpha}(\tilde{n}U_{R}+\Re_{1})\notag\\&-
\int\partial^{\alpha}\nabla\Phi_{R}\cdot[\partial^{\alpha},u]\Delta\Phi_{R}+\int\partial^{\alpha}\nabla\Phi_{R}\cdot(\nabla u\cdot\partial^{\alpha}\nabla\Phi_{R})
-\frac 12\int\mathrm{div}u|\partial^{\alpha}\nabla\Phi_{R}|^{2}\notag\\&-\varepsilon^{\frac{1}{2}}\int\partial^{\alpha}\nabla\Phi_{R}\cdot\partial^{\alpha}(u\Delta\phi^{(N)}).
\end{split}
\end{equation*}
Finally, by Young's inequality and the commutator estimate Remark \ref{R1} in the Appendix, we have
\begin{equation*}
\begin{split}
R_{2,7}\lesssim &-\frac{d}{dt}\|\partial^{\alpha}\nabla\Phi_{R}\|_{L^{2}}^{2}+\|\nabla\Phi_{R}\|_{H^{k}}^{2}+\varepsilon+
\varepsilon\|U_{R}\|_{H^{3}}^{2}+\varepsilon^{\frac{1}{2}}\|\nabla\Phi_{R}\|_{H^{k}}\|u\|_{H^{3}}\notag\\&+\|\nabla\Phi_{R}\|_{H^{k}}(\|\nabla u\|_{L^{\infty}}\|\Delta\Phi_{R}\|_{H^{k-1}}+\|u\|_{H^{k}}\|\Delta\Phi_{R}\|_{L^{\infty}})\notag\\&+\|(\nabla u,\mathrm{div}u)\|_{L^{\infty}}\|\nabla\Phi_{R}\|_{H^{k}}^{2}\\
\lesssim &-\frac{d}{dt}\|\partial^{\alpha}\nabla\Phi_{R}\|_{L^{2}}^{2}+(1+\varepsilon^{2N}\|U_{R}\|_{H^{3}}^{2})\|(U_{R},\nabla\Phi_{R})\|_{H^{3}}^{2}+\varepsilon.
\end{split}
\end{equation*}
Putting the estimates together, and taking $\delta=\frac 1{32}$, we complete the proof of lemma \ref{G2}.
\end{proof}
\begin{lemma}\label{G3}
Under the assumptions in Lemma \ref{G1}, we obtain
\begin{equation}\label{g3}
\begin{split}
&\frac{d}{dt}\|\partial^{\alpha}T_{R}\|_{L^{2}}^2+\kappa\|\partial^{\alpha}\nabla T_{R}\|_{L^{2}}^2\\
\lesssim &(1+\varepsilon^{6N}\|(U_{R},N_{R},T_{R})\|_{H^{3}}^{6}\|(N_{R},U_{R},T_{R},\hbar\nabla N_{R},\hbar\nabla U_{R})\|_{H^{3}}^2\\&+\frac {(\mu+\lambda)}{16}\|\mathrm{div}U_{R}\|_{H^{k}}^2+\frac {\mu}{32}\|\nabla U_{R}\|_{H^{k}}^2+\hbar^{4}\|\Delta U_{R}\|_{H^{k}}^2+\varepsilon.
\end{split}
\end{equation}
\end{lemma}
\begin{proof}
Applying the operator $\partial^{\alpha}$ to \eqref{rem1-3} and taking inner product with $\partial^{\alpha}T_{R}$ yield
\begin{equation}
\frac{}{}\begin{split}
\frac{1}{2}\frac{d}{dt}\|\partial^{\alpha}T_{R}\|_{L^{2}}^2= &\int\partial^{\alpha}G\partial^{\alpha}T_{R}-\int\partial^{\alpha}(u\cdot\nabla T_{R})\partial^{\alpha}T_{R}+
\frac{2\kappa}{3}\int \partial^{\alpha}(\frac{\Delta T_{R}}{n})\partial^{\alpha}T_{R}\notag\\&-\frac{\hbar^{2}}{36}\int\partial^{\alpha}\Big(\mathrm{div}\Delta U_{R}+\varepsilon^{N}\frac{\nabla N_{R}\Delta U_{R}}{n}\Big)\partial^{\alpha}T_{R}
= \sum_{i=1}^{4}R_{3,i}.
\end{split}
\end{equation}
Thanks to lemma \ref{L1}, \eqref{f1}, and \eqref{r1}$\sim$\eqref{r5} in the Appendix, $R_{3,1}$ can be estimated by
\begin{equation*}
\begin{split}
&2(\mu+\lambda)\delta\|\mathrm{div}U_{R}\|_{H^{k}}^{2}+\mu\delta\|\nabla U_{R}\|_{H^{k}}^{2}+\hbar^{4}\|\Delta U_{R}\|_{H^k}^{2}\notag\\&+(1+\varepsilon^{6N}\|(N_{R},U_{R},T_{R})\|_{H^{3}}^{6})\|(\hbar\nabla N_{R},\hbar\nabla U_{R},N_{R},T_{R},U_{R})\|_{H^{3}}^{2}+\varepsilon,
\end{split}
\end{equation*}
for any sufficiently small constant $\delta>0$. Using the commutator estimate and integration by parts, we obtain
\begin{equation*}
\begin{split}
R_{3,2}= &-\int[\partial^{\alpha},u]\nabla T_{R}\partial^{\alpha}T_{R}-\int u\partial^{\alpha}\nabla T_{R}\partial^{\alpha}T_{R}\\
= &-\int[\partial^{\alpha},u]\nabla T_{R}\partial^{\alpha}T_{R}+\frac 12\int \mathrm{div}u|\partial^{\alpha}T_{R}|^{2}\lesssim (1+\varepsilon^{N}\|U_{R}\|_{H^{3}})\|T_{R}\|_{H^{3}}^2.
\end{split}
\end{equation*}
Again, using the commutator estimate, H\"older inequality and integration by parts, we have
\begin{equation*}
\begin{split}
R_{3,3}= &\frac{2\kappa}{3}\int[\partial^{\alpha},\frac 1n]\Delta T_{R}\partial^{\alpha}T_{R}+\frac{2\kappa}{3}\int\frac{\partial^{\alpha}\Delta T_{R}\partial^{\alpha}T_{R}}{n}\\
= &\frac{2\kappa}{3}\int[\partial^{\alpha},\frac 1n]\Delta T_{R}\partial^{\alpha}T_{R}-\frac{2\kappa}{3}\int\frac{|\partial^{\alpha}\nabla T_{R}|^{2}}{n}+\frac{2\kappa}{3}\int\frac{\nabla n\cdot\partial^{\alpha}\nabla T_{R}\partial^{\alpha}T_{R}}{n^{2}}\\
\lesssim &-\frac{2\kappa}{3}\int\frac{|\partial^{\alpha}\nabla T_{R}|^{2}}{n}+\kappa\delta\|\nabla T_{R}\|_{H^{3}}^{2}+\frac 1\delta(1+\varepsilon^{6N}\|N_{R}\|_{H^{3}}^{6})\|T_{R}\|_{H^{3}}^{2}.
\end{split}
\end{equation*}
Next, we will estimate $R_{3,4}$. It can be estimated by integration by parts, Young's inequality, \eqref{r1} and \eqref{r2} in the Appendix that
\begin{equation*}
\begin{split}
R_{3,4}= &-\frac{\hbar^{2}}{36}\int\partial^{\alpha}\Big(\mathrm{div}\Delta U_{R}+\varepsilon^{N}\frac{\nabla N_{R}\Delta U_{R}}{n}\Big)\partial^{\alpha}T_{R}\\
\lesssim &\kappa\delta\|\nabla T_{R}\|_{H^{k}}^{2}+\frac{\hbar^{4}}{\kappa\delta}\|\Delta U_{R}\|_{H^{k}}^{2}\notag\\&+\varepsilon^{N}(\|\frac{\nabla N_{R}}{n}\|_{L^{\infty}}\|\Delta U_{R}\|_{H^k}+\|\frac 1n\|_{H^{3}}\|\nabla N_{R}\|_{H^{3}}\|\Delta U_{R}\|_{L^{\infty}})\|T_{R}\|_{H^{k}}\\
\lesssim &\kappa\delta\|\nabla T_{R}\|_{H^{k}}^{2}+2\hbar^{4}\|\Delta U_{R}\|_{H^k}^{2}\notag\\&+(1+\varepsilon^{4N}\|(N_{R},T_{R})\|_{H^{3}}^{4})\|(\hbar\nabla N_{R},\hbar\nabla U_{R},T_{R})\|_{H^{3}}^{2},
\end{split}
\end{equation*}
for any sufficiently small constant $\delta>0$.

We put the above estimates together, choose $\delta$ sufficiently small, say $\delta=\frac 1{32}$, and recall \eqref{f1} to finish the proof.
\end{proof}

However, the inequalities in Lemma $\ref{G1}\sim$ Lemma \ref{G3} are not closed. To close the inequality, we need to get suitable controls for $\|\hbar^{2}\Delta N_{R}\|_{H^{3}}^{2}$. This leads us to show the following lemma.
\begin{lemma}\label{G4}
Under the assumptions in Lemma \ref{G1}, we have
\begin{equation}\label{g4}
\begin{split}
&\frac{d}{dt}\|\partial^{\alpha}(\hbar\nabla U_{R},\hbar^{2}\Delta N_{R},\hbar\Delta\Phi_{R})\|_{L^{2}}^2+\mu\hbar^{2}\|\partial^{\alpha}\Delta U_{R}\|_{L^{2}}^2+(\mu+\lambda)\hbar^{2}\|\partial^{\alpha}\nabla\mathrm{div}U_{R}\|_{L^{2}}^2\\
\lesssim & \frac{\kappa}{32}\|\nabla T_{R}\|_{H^{k}}^2+\frac{(\mu+\lambda)}{32}\|\mathrm{div}U_{R}\|_{H^{k}}^2+(1+\varepsilon^{5N}\|(U_{R},T_{R},N_{R},\hbar\nabla N_{R})\|_{H^{3}}^{10})\\&
\|(U_{R},N_{R},T_{R},\hbar\nabla N_{R},\hbar\nabla U_{R},\hbar\mathrm{div}U_{R},\hbar\Delta\Phi_{R},\hbar^{2}\Delta N_{R})\|_{H^{3}}^2+\varepsilon.
\end{split}
\end{equation}
\end{lemma}
\begin{proof}
Applying the operator $\partial^{\alpha}$ to \eqref{rem1-2} and taking inner product with $-\hbar^{2}\partial^{\alpha}\Delta U_{R}$, we obtain
\begin{equation}
\begin{split}
&\frac{1}{2}\frac{d}{dt}\hbar^{2}\|\partial^{\alpha}\nabla U_{R}\|_{L^{2}}^2\\
=  &-\int\hbar^{2}\partial^{\alpha}F\cdot\partial^{\alpha}\Delta U_{R}+\hbar^{2}\int\partial^{\alpha}(u\cdot\nabla U_{R})\cdot \partial^{\alpha}\Delta U_{R}\notag\\&+\hbar^{2}\int\partial^{\alpha}(\frac{T\nabla N_{R}}{n})\cdot\partial^{\alpha}\Delta U_{R}-\frac{\hbar^{4}}{12}\int\partial^{\alpha}(\frac{\nabla\Delta N_{R}}{n})\cdot\partial^{\alpha}\Delta U_{R}\notag\\&-\mu\hbar^{2}\int\partial^{\alpha}(\frac{\Delta U_{R}}{n})\cdot\partial^{\alpha}\Delta U_{R}-(\mu+\lambda)\hbar^{2}\int\partial^{\alpha}(\frac{\nabla\mathrm{div}U_{R}}{n})\cdot\partial^{\alpha}\Delta U_{R}\notag\\&-\frac{1}{\varepsilon^{\frac{1}{2}}}\hbar^{2}\int\partial^{\alpha}\nabla\Phi_{R}\cdot \partial^{\alpha}\Delta U_{R}=\sum_{i=1}^{7}R_{4,i}.
\end{split}
\end{equation}
Similar  to $R_{2,1}$, the first term can be treated by Young's inequality, Lemma \ref{L2}, Lemma \ref{L3} and \eqref{r1}$\sim$\eqref{r5} in the Appendix,
\begin{equation*}
\begin{split}
R_{4,1}
\lesssim &\hbar^{4}\|\Delta U_{R}\|_{H^{k}}^{2}+\kappa\delta\|\nabla T_{R}\|_{H^{k}}^{2}+\frac{\hbar^{4}}{\kappa\delta}\|\Delta U_{R}\|_{H^{k}}^{2}\notag\\&+(1+\varepsilon^{10N}\|(N_{R},\hbar\nabla N_{R})\|_{H^{3}}^{10})\|(N_{R},T_{R},U_{R},\hbar\nabla N_{R},\hbar^{2}\Delta N_{R})\|_{H^{3}}^{2},
\end{split}
\end{equation*}for any sufficiently small constant $\delta>0$.
Next, by integration by parts twice and the commutator estimate, the second term $R_{4,2}$ on the right hand side can be bounded
\begin{equation*}
\begin{split}
R_{4,2}= &\hbar^{2}\int\partial^{\alpha}(u\cdot\nabla U_{R})\cdot \partial^{\alpha}\Delta U_{R}\\
= &\hbar^{2}\int[\partial^{\alpha},u]\nabla U_{R}\cdot\partial^{\alpha}\Delta U_{R}+\hbar^{2}\int u\cdot\partial^{\alpha}\nabla U_{R}\cdot\partial^{\alpha}\Delta U_{R}\\
= &\hbar^{2}\int[\partial^{\alpha},u]\nabla U_{R}\cdot\partial^{\alpha}\Delta U_{R}-\hbar^{2}\int \partial_{k}u^{i}\partial^{\alpha}\partial_{i} U_{R}^{j}\partial^{\alpha}\partial_{k} U_{R}^{j}+\frac{\hbar^{2}}2\int\mathrm{div}u|\partial^{\alpha}\nabla U_{R}|^{2}\\
\lesssim &\hbar^{4}\|\Delta U_{R}\|_{H^{k}}^{2}+(1+\varepsilon^{2N}\|U_{R}\|_{H^{3}}^{2})\|(U_{R},\hbar\nabla U_{R})\|_{H^{3}}^{2}.
\end{split}
\end{equation*}
For the third term, we apply integration by parts, the commutator estimate, and use Lemma \ref{L2} and Lemma \ref{L3} to get
\begin{equation*}
\begin{split}
R_{4,3}=& \hbar^{2}\int[\partial^{\alpha},\frac{T}{n}]\nabla N_{R}\cdot\partial^{\alpha}\Delta U_{R}+\hbar^{2}\int\frac{T}{n}\partial^{\alpha}\nabla N_{R}\cdot\partial^{\alpha}\Delta U_{R}\\
= &\hbar^{2}\int[\partial^{\alpha},\frac{T}{n}]\nabla N_{R}\cdot\partial^{\alpha}\Delta U_{R}\notag\\&-\hbar^{2}\int\nabla(\frac{T}{n}):\partial^{\alpha}\nabla U_{R}:\partial^{\alpha}\nabla N_{R}-\hbar^{2}\int\frac{T}{n}\partial^{\alpha}\nabla^{2} N_{R}:\partial^{\alpha}\nabla U_{R}\\
\lesssim &\mu\delta\|\nabla U_{R}\|_{H^{k}}^{2}+\hbar^{4}\|\Delta U_{R}\|_{H^{k}}^{2}+\frac{1}{\mu\delta}(1+\varepsilon^{2N}\|(N_{R},T_{R})\|_{H^{2}}^{2})\|\hbar^{2}\Delta N_{R}\|_{H^{k}}^{2}\notag\\&+(1+\varepsilon^{8N}\|(N_{R},T_{R})\|_{H^{3}}^{8})\|(N_{R},\hbar\nabla N_{R},\hbar\nabla U_{R})\|_{H^{k}}^{2},
\end{split}
\end{equation*}
for any sufficiently small constant $\delta>0$.

Similar to $R_{2,4}$, the estimate for $R_{4,4}$ is not straightforward since it involves higher order derivatives. We note that by \eqref{rem1-1} and the commutator, $R_{4,4}$ can be further decomposed into
\begin{equation*}
\begin{split}
R_{4,4}= &-\frac{\hbar^{4}}{12}\int\partial^{\alpha}(\frac{\nabla\Delta N_{R}}{n})\cdot\partial^{\alpha}\Delta U_{R}\\
= &-\frac{\hbar^{4}}{12}\int[\partial^{\alpha},\frac{1}{n}]\nabla\Delta N_{R}\cdot\partial^{\alpha}\Delta U_{R}-\frac{\hbar^{4}}{12}\int\frac{\partial^{\alpha}\nabla\Delta N_{R}\cdot\partial^{\alpha}\Delta U_{R}}{n}\\
= &-\frac{\hbar^{4}}{12}\int[\partial^{\alpha},\frac{1}{n}]\nabla\Delta N_{R}\cdot\partial^{\alpha}\Delta U_{R}\notag\\&
+\frac{\hbar^{4}}{12}\int\partial^{\alpha}\Delta N_{R}\nabla(\frac 1n)\cdot\partial^{\alpha}\Delta U_{R}+\frac{\hbar^{4}}{12}\int\frac{\partial^{\alpha}\Delta N_{R}\partial^{\alpha}\Delta\mathrm{div}U_{R}}n\\
= &-\frac{\hbar^{4}}{24}\frac{d}{dt}\int\frac{|\partial^{\alpha}\Delta N_{R}|^{2}}{n^{2}}-\frac{\hbar^{4}}{12}\int\frac{\partial^{\alpha}\Delta N_{R}\partial^{\alpha}\Delta(u\cdot\nabla N_{R})}{n^{2}}+R_{4,4,1},
\end{split}
\end{equation*}
where $R_{4,4,1}$ is denoted by
\begin{equation*}
\begin{split}
R_{4,4,1}
= &-\frac{\hbar^{4}}{12}\int[\partial^{\alpha},\frac{1}{n}]\nabla\Delta N_{R}\cdot\partial^{\alpha}\Delta U_{R}
+\frac{\hbar^{4}}{12}\int\partial^{\alpha}\Delta N_{R}\nabla(\frac 1n)\cdot\partial^{\alpha}\Delta U_{R}\notag\\&-\frac{\hbar^{4}}{12}\Big\{\int\frac{|\partial^{\alpha}\Delta N_{R}|^{2}}{n^{3}}\partial_{t}n+\varepsilon\int\frac{\partial^{\alpha}\Delta N_{R}\partial^{\alpha}\Delta(U_{R}\cdot\nabla\tilde{n}+N_{R}\mathrm{div}\tilde{u}+\Re_{1})}{n^{2}}\notag\\&+\int\frac{\partial^{\alpha}\Delta N_{R}[\partial^{\alpha}\Delta,n]\mathrm{div}U_{R}}{n^{2}}\Big\}.
\end{split}
\end{equation*}
By Lemma \ref{L1}, \eqref{f1} and Remark \ref{R1} in the Appendix, we derive
\begin{equation*}
\begin{split}
R_{4,4,1}
\lesssim &\hbar^{4}\|\nabla\mathrm{div} U_{R}\|_{H^{k}}^{2}+\hbar^{4}\|\Delta U_{R}\|_{H^{k}}^{2}\notag\\&+(1+\varepsilon^{3N}\|(N_{R},U_{R})\|_{H^{3}}^{6})\|(N_{R},U_{R},\hbar\nabla U_{R},\hbar\nabla N_{R},\hbar^{2}\Delta N_{R})\|_{H^{3}}^{2}+\varepsilon.
\end{split}
\end{equation*}
For the second term, from the commutator estimate and integration by part two times, we arrive at
\begin{equation*}
\begin{split}
&-\frac{\hbar^{4}}{12}\int\frac{\partial^{\alpha}\Delta N_{R}\partial^{\alpha}\Delta(u\cdot\nabla N_{R})}{n^{2}}\\= &-\frac{\hbar^{4}}{12}\int\frac{\partial^{\alpha}\Delta N_{R}[\partial^{\alpha}\Delta,u]\cdot\nabla N_{R}}{n^{2}}-\frac{\hbar^{4}}{12}\int\frac{\partial^{\alpha}\Delta N_{R}u\cdot\partial^{\alpha}(\Delta\nabla N_{R})}{n^{2}}\\
= &-\frac{\hbar^{4}}{12}\int\frac{\partial^{\alpha}\Delta N_{R}[\partial^{\alpha}\Delta,u]\cdot\nabla N_{R}}{n^{2}}+\frac{\hbar^{4}}{24}\int\mathrm{div}(\frac{u}{n^{2}})|\partial^{\alpha}\Delta N_{R}|^{2}\\
\lesssim &(1+\varepsilon^{2N}\|(N_{R},U_{R})\|_{H^{3}}^{2})\|(N_{R},\hbar^{2}\Delta N_{R})\|_{H^{k}}^{2}+\hbar^{4}\|\Delta U_{R}\|_{H^{k}}^{2}.
\end{split}
\end{equation*}
Summarizing, we have
\begin{equation*}
\begin{split}
R_{4,4}\lesssim &-\frac{\hbar^{4}}{24}\frac{d}{dt}\int\frac{|\partial^{\alpha}\Delta N_{R}|^{2}}{n^{2}}+\hbar^{4}\|\Delta U_{R}\|_{H^{k}}^{2}+\varepsilon\notag\\&+(1+\varepsilon^{3N}\|(N_{R},U_{R})\|_{H^{3}}^{6})\|(N_{R},U_{R},\hbar\nabla U_{R},\hbar\nabla N_{R},\hbar^{2}\Delta N_{R})\|_{H^{3}}^{2}.
\end{split}
\end{equation*}
Now we turn to estimate $R_{4,5}$ and $R_{4.6}$. Using integration by parts and the commutator estimate, we observe
\begin{equation*}
\begin{split}
R_{4,5}= &-\hbar^{2}\mu\int\partial^{\alpha}(\frac{\Delta U_{R}}{n})\cdot\partial^{\alpha}\Delta U_{R}\\
= &-\hbar^{2}\mu\int[\partial^{\alpha},\frac 1n]\Delta U_{R}\cdot\partial^{\alpha}\Delta U_{R}-\hbar^{2}\mu\int\frac{|\partial^{\alpha}\Delta U_{R}|^{2}}{n}\\
\lesssim &-\hbar^{2}\mu\int\frac{|\partial^{\alpha}\Delta U_{R}|^{2}}{n}+\hbar^{2}\mu\delta\|\Delta U_{R}\|_{H^{k}}^{2}+\frac{1}{\delta}(1+\varepsilon^{6N}\|N_{R}\|_{H^{3}}^{6})\|\hbar\nabla U_{R}\|_{H^{3}}^{2}.
\end{split}
\end{equation*}
For the estimate of $R_{4,6}$, we apply integration by part twice to obtain
\begin{equation*}
\begin{split}
R_{4,6}= &-\hbar^{2}(\mu+\lambda)\int\partial^{\alpha}(\frac{\nabla\mathrm{div}U_{R}}{n})\cdot\partial^{\alpha}\Delta U_{R}\\
= &-\hbar^{2}(\mu+\lambda)\int[\partial^{\alpha},\frac 1n]\nabla\mathrm{div}U_{R}\cdot\partial^{\alpha}\Delta U_{R}
-\hbar^{2}(\mu+\lambda)\int\frac{|\partial^{\alpha}\nabla\mathrm{div}U_{R}|^{2}}{n}\notag\\&
-\hbar^{2}(\mu+\lambda)\Big\{\int\frac{\partial_{k}n\partial^{\alpha}\partial_{j}\mathrm{div}U_{R}\partial^{\alpha}\partial_{k}U_{R}^{j}}{n^{2}}
-\int\frac{\partial_{j}n\partial^{\alpha}\partial_{k}\mathrm{div}U_{R}\partial^{\alpha}\partial_{k}U_{R}^{j}}{n^{2}}\Big\}
\\
\lesssim &-\hbar^{2}(\mu+\lambda)\int\frac{|\partial^{\alpha}\nabla\mathrm{div}U_{R}|^{2}}{n}
+\hbar^{2}(\mu+\lambda)\delta\|\nabla\mathrm{div}U_{R}\|_{H^{k}}^{2}\notag\\&+\hbar^{2}(\mu+\lambda)\mu\delta\|\Delta U_{R}\|_{H^{k}}^{2}+\frac{1}{\delta}(1+\varepsilon^{2N}\|N_{R}\|_{H^{3}}^{2})\|\hbar\nabla U_{R}\|_{H^{3}}^{2}\notag\\&
+\frac{1}{\mu\delta}(1+\varepsilon^{6N}\|N_{R}\|_{H^{3}}^{6})\|\hbar\mathrm{div}U_{R}\|_{H^{3}}^{2}.
\end{split}
\end{equation*}
In what follows, thanks to \eqref{rem1-1} and \eqref{rem1-4}, we divide $R_{4,7}$ into following
\begin{equation*}
\begin{split}
R_{4,7}= &-\varepsilon^{-\frac{1}{2}}\hbar^{2}\int\partial^{\alpha}\nabla\Phi_{R}\cdot\partial^{\alpha}\Delta U_{R}
= \varepsilon^{-\frac{1}{2}}\hbar^{2}\int\partial^{\alpha}\Delta\Phi_{R}\partial^{\alpha}\mathrm{div}U_{R}\\
= &-\varepsilon^{-\frac{1}{2}}\hbar^{2}\int\frac{\partial^{\alpha}\Delta\Phi_{R}[\partial^{\alpha},n]\mathrm{div}U_{R}}{n}\notag\\
&-\varepsilon^{-\frac{1}{2}}\hbar^{2}\int\frac{\partial^{\alpha}\Delta\Phi_{R}\partial^{\alpha}(\partial_{t}N_{R}+u\cdot\nabla N_{R}+\varepsilon(\nabla\tilde{n}\cdot U_{R}+\mathrm{div}\tilde{u}N_{R}+\Re_{1}))}{n}\\
= &-\frac{\hbar^{2}}2\frac{d}{dt}\int\frac{|\partial^{\alpha}\Delta\Phi_{R}|^{2}}{n}-\varepsilon^{-\frac{1}{2}}\hbar^{2}\int\frac{\partial^{\alpha}\Delta\Phi_{R}\partial^{\alpha}(u\cdot\nabla N_{R})}{n}+R_{4,7,1},
\end{split}
\end{equation*}
where we denote $R_{4,7,1}$ as
\begin{equation*}
\begin{split}
R_{4,7,1}
= &-\varepsilon^{-\frac{1}{2}}\hbar^{2}\int\frac{\partial^{\alpha}\Delta\Phi_{R}[\partial^{\alpha},n]\mathrm{div}U_{R}}{n}\notag\\&+\frac{\hbar^{2}}2\int|\partial^{\alpha}\Delta\Phi_{R}|^{2}\partial_{t}(\frac 1n)-\hbar^{2}\varepsilon^{\frac{1}{2}}\int\frac{\partial^{\alpha}\Delta\Phi_{R}\partial^{\alpha}\Delta\partial_{t}\phi^{(N)}}{n}\notag\\&-\varepsilon^{\frac{1}{2}}\hbar^{2}\int\frac{\partial^{\alpha}\Delta\Phi_{R}\partial^{\alpha}(\nabla\tilde{n}\cdot U_{R}+\mathrm{div}\tilde{u}N_{R}+\Re_{1})}{n}.
\end{split}
\end{equation*}
For the first one $R_{4,7,1}$, since $\phi^{N}$ is known and is assumed to be as smooth as we want, thanks to \eqref{rem1-4}, Lemma \ref{L1}, Lemma \ref{L2}, Young's inequality and Remark \ref{R1} in the Appendix, we have
\begin{equation*}
\begin{split}
R_{4,7,1}
\lesssim &\hbar^{2}\|\Delta\Phi_{R}\|_{H^{k}}^{2}+(1+\varepsilon^{2N-1}\|(U_{R},N_{R})\|_{H^{3}}^{2})\|(N_{R},U_{R},\hbar\Delta\Phi_{R})\|_{H^{3}}^{2}+\varepsilon,
\end{split}
\end{equation*}
where we assume $N\geq 1$, such that $\varepsilon^{2N-1}\leq\varepsilon^{N}$.
For the second term of $R_{4,7}$, thanks to the commutator estimate, Sobolev embedding and \eqref{rem1-4}, we have
\begin{equation*}
\begin{split}
&-\hbar^{2}\frac{1}{\varepsilon^{\frac{1}{2}}}\int\frac{\partial^{\alpha}\Delta\Phi_{R}[\partial^{\alpha},u]\cdot\nabla N_{R}}{n}-\hbar^{2}\frac{1}{\varepsilon^{\frac{1}{2}}}\int\frac{\partial^{\alpha}\Delta\Phi_{R} u\cdot\partial^{\alpha}\nabla N_{R}}{n}\\
\lesssim &\hbar^{2}\|\Delta\Phi_{R}\|_{H^{k}}\Big\{\|\nabla u\|_{L^{\infty}}(\|\nabla\Delta\Phi_{R}\|_{H^{k-1}}+\varepsilon^{\frac 12})+\|u\|_{H^{k}}(\|\nabla\Delta\Phi_{R}\|_{L^{\infty}}+\varepsilon^{\frac 12})\Big\}\notag\\&+\frac{\hbar^{2}}2\int\mathrm{div}(\frac un)|\partial^{\alpha}\Delta\Phi_{R}|^{2}-\hbar^{2}\varepsilon^{\frac 12}\int\frac{\partial^{\alpha}\Delta\Phi_{R}u\cdot\partial^{\alpha}\nabla\Delta\phi^{(N)}}{n}\\
\lesssim &\hbar^{2}(1+\varepsilon^{2N}\|(U_{R},N_{R})\|_{H^{3}}^{2})\|\Delta\Phi_{R}\|_{H^{3}}^{2}+\varepsilon.
\end{split}
\end{equation*}
Summing up the estimates for $R_{4,7}$ to derive
\begin{equation*}
\begin{split}
R_{4,7}\lesssim &-\frac{\hbar^{2}}2\frac{d}{dt}\int\frac{|\partial^{\alpha}\Delta\Phi_{R}|^{2}}{n}+(1+\varepsilon^{N}\|(U_{R},N_{R})\|_{H^{3}}^{2})\|(N_{R},U_{R},\hbar\Delta\Phi_{R})\|_{H^{3}}^{2}+\varepsilon.
\end{split}
\end{equation*}
Now, taking $\delta=\frac 1{32}$ and $\hbar\ll 1$ such that $\hbar^{4}\ll\hbar^{2}$, then we can complete the proof of Lemma \ref{G4}.
\end{proof}

Then integrating these estimates from Lemma \ref{G1} to Lemma \ref{G4} over $[0,t]\subseteq[0,\min\{\tau_{\varepsilon},\tau\}]$, and summing them up for all multi-index $\alpha$ with $0\leq|\alpha|\leq 3$, we obtain Proposition \ref{prop1}.
\section{Proof of Theorem \ref{thm3}}\label{proof}
\begin{proof}
Thanks to \eqref{def-A}, we obtain the following Gronwall type inequality
\begin{equation*}
\begin{split}
\frac{d}{dt}|\!|\!|(N_{R},U_{R},T_{R},\nabla\Phi_{R})|\!|\!|_{3}^2\leq & C_{1}(1+\varepsilon^{5N}\|(N_{R},U_{R},T_{R})\|_{H^{3}}^{10})|\!|\!|(N_{R},U_{R},T_{R},\nabla\Phi_{R})|\!|\!|_{3}^2\\&+\varepsilon.
\end{split}
\end{equation*}
Next, there exists $\varepsilon_{1}>0$ such that for any $0<\varepsilon<\varepsilon_{1}$, we have $$\varepsilon^{5N}\|(N_{R},U_{R},T_{R})\|_{H^{3}}^{10}<1,$$
thanks to \eqref{prior}. Moreover, we derive $$\frac{d}{dt}|\!|\!|(N_{R},U_{R},T_{R},\nabla\Phi_{R})|\!|\!|_{H^{3}}^2\leq 2C_{1}|\!|\!|(N_{R},U_{R},T_{R},\nabla\Phi_{R})|\!|\!|_{3}^2+\varepsilon.$$
Let $C_{0}=|\!|\!|(N_{R},U_{R},T_{R},\nabla\Phi_{R})(0)|\!|\!|_{3}^2$, hence we get the inequality from the Gronwall type inequality,
$$|\!|\!|(N_{R},U_{R},T_{R},\nabla\Phi_{R})|\!|\!|_{3}^2\leq e^{2C_{1}t}(C_{0}+\varepsilon t).$$

For arbitrarily given $0<\tau<\tau_*$ in Theorem \ref{thm1}, there exists $\varepsilon_{2}>0$, such that $\varepsilon t<1$ for any $0<\varepsilon<\varepsilon_{2}$. If we chose $\tilde{C}$ sufficiently large such that $\tilde{C}\geq e^{2C_{1}\tau}(C_{0}+1),$ then
$$|\!|\!|(N_{R},U_{R},T_{R},\nabla\Phi_R)|\!|\!|_{3}^2\leq\tilde{C}.$$ The proof is then complete by classical continuity method.
\end{proof}

\bigskip

\section{Appendix}
\begin{remark}\label{R1}
Let $\alpha$ be any multi-index with $|\alpha|=k$, $k\geq 1$ and $p\in(1,\infty)$. Then there holds
\begin{equation}\label{r1}
\begin{split}
\|\partial_{x}(fg)\|_{L^{p}}\lesssim &\|f\|_{L^{p_{1}}}\|g\|_{H^{k,p_{2}}} +\|f\|_{H^{k,p_{3}}}\|g\|_{L^{p_{4}}}\\
\|[\partial_{x},f]g\|_{L^{p}}\lesssim &\|\nabla f\|_{L^{p_{1}}}\|g\|_{H^{k-1,p_{2}}} +\|f\|_{H^{k,p_{3}}}\|g\|_{L^{p_{4}}},
\end{split}
\end{equation}
where $f,g\in \mathbb{S}$, the Schwartz class and $p_{2},p_{3}\in(1,\infty)$ such that $\frac 1p=\frac 1{p_{1}}+\frac 1{p_{2}}=\frac 1{p_{3}}+\frac 1{p_{4}}$.
\end{remark}
Moreover, due to Sobolev embedding and Remark \ref{R1} in the Appendix, then in particular, we can obtain with $k\geq2$,
\begin{equation}\label{r2}
\begin{split}
\|f\|_{H^{k}}\lesssim &\|f\|_{L^{\infty}}\|g\|_{H^{k}} +\|f\|_{H^{k}}\|g\|_{L^{\infty}} \lesssim \|f\|_{H^{k}}\|g\|_{H^{k}}.
\end{split}
\end{equation}
Hence, it is obvious that, for any $m\geq1$, together with $\frac 12\leq n\leq\frac 23$, we can deduce the following inequality
\begin{equation}\label{r3}
\begin{split}
\|(\frac 1n)^{m}\|_{H^{3}}\lesssim &\|(\frac 1n)^{m-1}\|_{L^{\infty}}\|\frac 1n\|_{H^{3}}+\|(\frac 1n)^{m-1}\|_{H^{3}}\|\frac 1n\|_{L^{\infty}}\\
\lesssim &\|\frac 1n\|_{H^{3}}+\|(\frac 1n)^{m-1}\|_{H^{3}}
\lesssim \cdots\cdots
\lesssim \varepsilon+\varepsilon^{3N}\|N_{R}\|_{H^{3}}^{3}.
\end{split}
\end{equation}
Similarly, we derive
\begin{equation}\label{r5}
\begin{split}
\|(\nabla N_{R})^{m}\|_{H^{k}}\lesssim &\|(\nabla N_{R})^{m-1}\|_{H^{k}}\|\nabla N_{R}\|_{L^{\infty}}+\|\nabla N_{R}\|_{L^{\infty}}^{m-1}\|\nabla N_{R}\|_{H^{k}}\\
\lesssim &\cdots\cdots\lesssim \|N_{R}\|_{H^{3}}^{m-1}\|\nabla N_{R}\|_{H^{k}}.
\end{split}
\end{equation}
\section*{Acknowledgments}
We would like to thank the anonymous reviewers for their valuable suggestions and fruitful
comments which led to significant improvement this work.

%\end{CJK*}

\begin{thebibliography}{99}

\bibitem{Bohm52} D. Bohm, A suggested interpretation of the quantum theory in terms of ``hidden" valuables: I; II, Phys. Rev., 85, (1952)166-179; 180-193.

\bibitem{BDD05} D. Bresch, B. Desjardins and B. Ducomet, Quasi-neutral limit for a viscous capillary model of plasma, Ann. Inst. H. Poincar\'{e} Anal. Nonlinear, 22, (2005)1-9.

\bibitem{BM10} S. Brull and F. M\'{e}hats, Derivation of viscous correction terms for the isothermal quantum Euler model, Z. Angew. Math. Mech., 90, (2010)219-230.

\bibitem{CF88} P. Constantin, C. Foias, Navier-Stokes Equations, University of Chicago Press, 1988.

\bibitem{CG00} S. Cordier and E. Grenier, Quasineutral limit of an Euler-Poisson system arising from plasma physics, Commun. Partial Differential Equations, 23, (2000)1099-1113.

\bibitem{CDM13} L. Chen, D. Donatelli and P. Marcati, Incompressible type limit analysis of a hydrodynamic model for charge-carrier transport. SIAM Journal on Mathematical Analysis, 45(3), (2013)915-933.

\bibitem{DGM07b} P. Degond, S. Gallego, and F. M\'{e}hats, Isothermal quantum hydrodynamics: derivation, asymptotic analysis, and simulation, Multiscale Model. Simul., 6(2007), 246-272.

\bibitem{DGMR08} P. Degond, S. Gallego, F. M\'{e}hats and C. Ringhofer, Quantum hydrodynamic and diffusion models derived from the entropy principle, in: G. Allaire et al. (eds.), Quantum transport, pp. 111-168, Lecture Notes Math. 1946. Springer, Berlin, 2008.

\bibitem{DGMR06} P. Degond, S. Gallego, F. M\'{e}hats and C. Ringhofer, Quantum diffusion models derived from the entropy principle, in: L. Bonilla et al. (eds.), Progress in Industrial Mathematics at ECMI 2006, pp. 106-122, Mathematics in Industry 12. Springer, Berlin, 2008.

\bibitem{DMR05} P. Degond, F. M\'{e}hats, and C. Ringhofer, Quantum energy-transport and drift-diffusion models, J. Stat. Phys., 118(2005), 625-665.

\bibitem{DR03} P. Degond and C. Ringhofer, Quantum moment hydrodynamics and the entropy principle, J. Stat. Phys., 112, (2003)587-628.

\bibitem{DM08} D. Donatelli and P. Marcati, A quasineutral type limit for the Navier-Stokes-Poisson system with large data, Nonlinearity, 21, (2008)135-148.

\bibitem{DM12} D. Donatelli and P. Marcati, Analysis of oscillations and defect measures for the quasineutral limit in plasma physics, Arch. Ration. Mech. Anal., 206, (2012)159-188.

\bibitem{DM15} D. Donatelli and P. Marcati, Quasineutral limit, dispersion and oscillations for Korteweg type fluids, SIAM J. Math. Anal., 47, (2015), 2265-2282.

\bibitem{DS85} J.E. Dunn and J. Serrin, On the thermodynamics of interstitial working, Arch. Ration. Mech. Anal., 88, (1985)95-133.

\bibitem{Gardner94} C.L. Gardner, The quantum hydrodynamic model for semiconductor devices, SIAM J. Appl. Math., 54(2), (1994)409-427.

\bibitem{GM01} I. Gasser and P. Marcati, The combined relaxation and vanishing Debye length limit in the hydrodynamic model for semiconductors, Math. Methods Appl. Sci., 24, (2001)81-92.

\bibitem{GM2001} I. Gasser and P. Marcati, A vanishing Debye length limit in a hydrodynamic model for semiconductors, in Hyperbolic Problems: Theory, Numerics, Applications, Vol. I, Internat. Ser. Numer. Math. 140, Birkh¨auser, Basel, (2001)409-414.

\bibitem{Haas11} F. Haas, Quantum plasmas: An hydrodynamic approach, Springer, New York, 2011.

\bibitem{HL94} H. Hattori and D. Li, Solutions for two dimensional system for materials of Korteweg type, SIAM J. Math. Anal., 25, (1994)85-98.

\bibitem{HL96} H. Hattori and D. Li,  Global solutions of a high dimensional system for Korteweg materials. J. Math. Anal. Appl., 198, (1996)84-97.

\bibitem{JLL09} Q. Ju, F. Li and H. Li, The quasineutral limit of compressible Navier-Stokes-Poisson system with heat conductivity and general initial data. J. Differential Equations 247, (2009)203-224.

\bibitem{Jungel10} A. J\"{u}ngel, Global weak solutions to compressible Navier-Stokes equations for quantum fluids, SIAM J. Math. Anal., 42(3), (2010)1025-1045.

\bibitem{JLW14} A. J\"{u}ngel, C.-K. Lin and K.-C. Wu, An asymptotic limit of a Navier-Stokes system with capillary effects, Comm. Math. Phys., 329, (2014)725-744.

\bibitem{JM11} A. J\"{u}ngel and J.-P. Mili\^{s}i\'{c}, Full compressible Navier-Stokes equations for quantum fluids: Derivation and numerical solutions, Kinetic and Related Models, 4(3), (2011)785-807.

\bibitem{Lions96} P.-L. Lions, Mathematical Topics in Fluid Mechanics, vol. 1: Incompressible Models, Oxford Lecture Ser. Math. Appl., vol. 3, The Clarendon Press/Oxford University Press, New York, 1996.

\bibitem{LL05} H. Li and C. Lin, Zero Debye length asymptotic of the quantum hydrodynamic model for semiconductors, Comm. Math. Phys., 256(1), (2005)195-212.

\bibitem{LY14} Y. Li and W. Yong, Quasi-neutral limit in a 3D compressible Navier-Stokes-Poisson-Korteweg model, IMA J. Appl. Math., (2014)1-16

\bibitem{MN80} A. Matsumura, T. Nishida, The initial value problem for the equations of motion of viscous and heat-conductive gases, J. Math. Kyoto Univ., 20-1, (1980)67-104.

\bibitem{PGG15} X. Pu and B. Guo, Global existence and semiclassical limit for quantum hydrodynamic equations with viscosity and heat conduction, , Kinetic \& Related Models, 9, (2016)165-191.

\bibitem{PGQ15} X. Pu and B. Guo, Quasineutral limit of the Euler-Poisson equation for a cold, ion-acoustic plasma, arXiv:1304.0187. To appear in Quart. Appl. Math..

\bibitem{PWY06} Y. Peng, Y. Wang and W. Yong, Quasi-neutral limit of the non-isentropic Euler-Poisson system, Proceedings of the Royal Society of Edinburgh, 136A, (2006) 1013-1026

\bibitem{Stein1970} E.M. Stein, Singular integrals and differentiability properties of functions, Princeton university press, 1970.

\bibitem{Temam01} R. Temam, Navier-Stokes Equations: Theory and Numerical Analysis, American Mathematical Society, 2001.

\bibitem{Wang04} S. Wang, Quasineutral limit of Euler-Poisson system with and without viscosity, Comm. Partial Differential Equations 29, (2004)419-456.

\bibitem{WJ06} S. Wang and S. Jiang, The convergence of Navier-Stokes-Poisson system to the incompressible Euler equations. Comm. Partial Differential Equations, 31, (2006)571-591.

\bibitem{Wigner32} E. Wigner, On the quantum correction for thermodynamic equilibrium, Phys. Rev., 40, (1932)749-759.

\end{thebibliography}
\end{document}